%% file: main.tex
\documentclass[11pt,letterpaper]{article}
\usepackage[in]{fullpage}  
\usepackage{float}
\usepackage{mathpazo}

\usepackage{amsmath}
\usepackage{amssymb}
\usepackage{mathtools}
\usepackage{amsthm}
\usepackage{wrapfig}
\usepackage{thm-restate,thmtools}
\usepackage{algorithm}
\usepackage{algorithmic}
\usepackage{hyperref}
\usepackage{booktabs}

\usepackage{url}
\usepackage{subcaption} 
\usepackage{multirow}
\usepackage{enumitem}
\usepackage{graphicx}

\newtheorem{theorem}{Theorem}
\newtheorem{lemma}{Lemma}

\newtheorem{definition}{Definition}
\newtheorem{assumption}{Assumption}
\newtheorem{remark}{Remark}

\title{A Unified Attack Framework against Statistical Generative Watermarks for Large Language Models}

\author{
  Romina Omidi\\
   Illinois Institute of Technology\\
 \texttt{romidi@hawk.illinoistech.edu} 
  \and
  Nima Naderloui \\
   University of Connecticut\\
   \texttt{nima.naderloui@uconn.edu}
  \and
  Yuan Hong \\
  University of Connecticut\\
  \texttt{yuan.hong@uconn.edu}
  \and
  Binghui Wang\\
  Illinois Institute of Technology\\
  \texttt{bwang70@illinoistech.edu} 
}

\date{}

\begin{document}
\maketitle

\input{Abstract}
\input{Introduction}

\input{Background}
\input{AttackFramework}

\input{Optimality}
\input{Tradeoff}
\input{TheoreticalEffectiveness}
\input{Experiments}
\input{RelatedWork}

\input{Conclusion}

\bibliographystyle{plainurl}
\bibliography{References}

\appendix
\clearpage
\input{Appendix}
\end{document}

%% file: Abstract.tex
\begin{abstract}

Watermarking has emerged as a key tool for identifying text generated by large language models (LLMs), with current schemes embedding statistical signals into the generation process to enable later detection. In this work, we identify a shared weakness across a broad class of popular LLM watermarking methods, including KGW (and its variants), GumbelSoft,  DipMark, ITS, and SynthID. We show that these schemes exhibit a fundamental link between their detection scores and the token probability shifts they introduce during generation. This link, which we call \emph{score-ratio alignment}, allows an attacker with only \emph{black-box query} access to the model—and no knowledge of secret keys or model weights—to produce text that evades detection while remaining fluent and natural. We prove that the resulting attack is optimal across all covered watermark families, and show that the weakness is intrinsic: any watermark designed to maximize statistical detection power necessarily induces the very structure that attackers can exploit. We evaluate our attack on seven popular generative LLM watermarking schemes, three 
LLMs, and tasks, which 
reduce detection rates to at most 0.05 across all settings while preserving text quality. Our findings show that current watermarking methods can be reliably defeated under realistic attack conditions, highlighting the need for fundamentally different watermarking designs that decouple detection scores from token probability shifts.
Code is available at \url{https://github.com/romidi80/Unified-Attack-Framework-for-LLM-Watermark}. 
\end{abstract}

%% file: Introduction.tex
\section{Introduction}
\label{sec:introduction}

The widespread deployment of large language models (LLMs)
capable of generating fluent, human-like text at scale has
created a pressing need for reliable provenance
attribution~\cite{brown2020language,
openai2024gpt4technicalreport,touvron2023llama2openfoundation,bai2023qwen}.
Generative watermarking has emerged as the most practical solution:
a statistical signal is embedded into the LLM's token
generation process, enabling a verifier to classify text
as machine-generated without access to the original prompt
or model weights~\cite{aaronson2022watermarking,kirchenbauer2023watermark,
christ2024undetectable,kuditipudi2024robust,dathathri2024scalable}.
Deployed and proposed schemes---including
KGW~\cite{kirchenbauer2023watermark},
GumbelSoft~\cite{fu2024gumbelsoft,aaronson2022watermarking},
DipMark~\cite{wu2024dipmark},
ITS~\cite{kuditipudi2024robust}, and
SynthID~\cite{dathathri2024scalable}---share a common
architecture: a pseudorandom function (PRF) keyed by a
secret key biases token selection at each decoding step,
and a detector aggregates per-token scores into a test
statistic compared against a calibrated threshold.

Despite significant advances in watermark design, a
central security question remains unresolved: \emph{can
an adversary with only black-box access to the
watermarked LLM---and no knowledge of the PRF key or
model weights---reliably evade detection while generating
fluent, semantically coherent text?}
Prior attacks have addressed a weaker version of this
question by applying post-hoc paraphrasing: a secondary
model rewrites the watermarked output to disrupt the
embedded signal~\cite{krishna2023paraphrasing,
Piet2023MARKMyWORDSAA,Sadasivan2023CanAT}.
Paraphrasing-based attacks impose a fundamental
quality--evasion tradeoff---aggressive rewriting
degrades fluency and semantic fidelity in proportion to
the evasion gain~\cite{krishna2023paraphrasing}---and
require inference through a capable auxiliary model.
Crucially, reported robustness improvements against
Paraphrasing does not characterise security against
generation-time adversaries.

\vspace{+0.05in} 
\noindent 
\textbf{Our approach.}
We identify a structural property shared by major
watermarking schemes, which we term \emph{score-ratio
alignment}: 
the per-token
watermark score is a monotone function of the probability
ratio between the watermarked and base distributions.
Formally, for each token $j$ from a vocabulary $\mathcal{V}$ we have: 
$  \arg\min_{j \in \mathcal{V}} \frac{\hat{p}_j}{p_j}
  \;\subseteq\;
  \arg\min_{j \in \mathcal{V}} \mathbb{E}[s_j]$, 
where $\hat{p}_j$ is the watermarked token probability,
$p_j$ the base model probability, and $\mathbb{E}[s_j]$ the
expected per-token detection score. 
We prove that the equation 
holds as a formal
theorem---not merely as an empirical observation---for
KGW (and its two variants), GumbelSoft, DipMark, ITS, and SynthID 
(see Theorems in \S\ref{sec:optimality}).

Score-ratio alignment admits a direct and efficient
attack.
An adversary who estimates the probability ratio
$\hat{p}_j/p_j$ from repeated black-box queries can
greedily select the score-minimizing token at each
decoding step, systematically suppressing the cumulative
detection statistic.
Because the selected tokens are drawn from the
low-distortion region of the base model distribution,
text quality is preserved.
The attack requires neither the secret key nor model weights, relying only on sampling access to the watermarked model and an offline surrogate model.

\vspace{+0.05in} 
\noindent {\bf Key insight.}
Score-ratio alignment is not an artifact of any
particular design; it is a structural consequence
of optimization for detection power.
By the Neyman--Pearson
lemma~\cite{Neyman1933OnTP}, any statistically optimal
detector must be based on the likelihood ratio
$\hat{p}(\mathbf{x})/p(\mathbf{x})$---exactly the
quantity that score-ratio alignment ties to the watermark
score.
We prove (Theorem~\ref{subsec:fundamental-tradeoff}, \S\ref{sec:tradeoff}) that
\emph{no watermarking scheme within the current
generator--detector paradigm can simultaneously achieve
Neyman--Pearson-optimal detection power and robustness
against a score-ratio-minimizing adversary.}
This result is not a limitation of specific parameter
setting or scheme; it is an
information-theoretic impossibility that holds for any
scheme whose score function is optimized for
distinguishability.

\vspace{+0.05in}
\noindent {\bf Contributions.}
This paper makes below contributions.

\begin{enumerate}[leftmargin=*, label=\textbf{(\arabic*)}, nosep]

\item \textbf{Score-ratio alignment and a unified attack.}
We identify score-ratio alignment as a shared 
property of all major LLM watermarking schemes and
formulate score-ratio minimization as a principled
black-box attack 
requiring only output tokens from 
the watermarked model (\S\ref{sec:attack}).

\item \textbf{Formal optimality.}
We prove that score-ratio minimization achieves the
minimum attainable expected detection score for a wide range of LLM watermarking schemes: 
no black-box adversary can achieve a
lower expected score within each scheme's construction
(\S\ref{sec:optimality}).

\item \textbf{Detectability--robustness tradeoff.}
We prove that optimizing the LLM watermark 
detection power establishes an inherent tension between detectability and
 robustness (\S\ref{sec:tradeoff}).

\item \textbf{Theoretical attack effectiveness.}
We theoretically characterize the effectiveness of our practical attack  by showing its true positive rate admits a small upper bound (\S\ref{sec:attack-effectiveness}).

\item \textbf{Empirical evaluation.}
We evaluate the attack on seven watermarking schemes,  
{three} LLMs,  
and
three text corpora (C4, FakeNews, and Eli5~\cite{Raffel2019ExploringTL}).
The attack reduces the true positive rate to at most
$0.05$ across all settings while preserving text quality 
(\S\ref{sec:experiments}).

\end{enumerate}

%% file: Background.tex
\begin{table*}[t]
  \centering
  \caption{Summary of the seven watermarking schemes
           covered by our framework.
           ``Distortional'' (D) schemes modify the
           sampling distribution; ``distortion-free'' (DF)
           schemes preserve it in expectation.
           $\gamma$ denotes the green-list fraction,
           $\delta$ the logit boost, $\xi_t$ the pseudorandom vector,
         and $\gamma_{\mathrm{EM}}$ the
           Euler--Mascheroni constant.}
  \label{tab:scheme-summary}
  \resizebox{\textwidth}{!}{%
  \begin{tabular}{llllll}
    \toprule
    \textbf{Scheme} & \textbf{Type} & \textbf{PRF output} $\mathcal{E}$
      & \textbf{Decoder} $\Gamma$
      & \textbf{Token score} $\phi(\xi_t, x_t)$
      & \textbf{Detection statistic} $\Phi$ \\
    \midrule
    KGW~\cite{kirchenbauer2023watermark}
      & D
      & $\{0,1\}^{|\mathcal{V}|}$
      & $\operatorname{softmax}(\mathbf{l}_t + \delta\xi_t; \tau)$
      & $\xi_t[x_t] \in \{0,1\}$
      & $Z$-score~\eqref{eq:kgw-stat} \\
      Token-Specific KGW~\cite{huo2024tokenspecific}
      & D
      & $\{0,1\}^{|\mathcal{V}|}$ (token-dep.)
      & $\operatorname{softmax}((\mathbf{l}_t{+}\delta_t{\odot}\xi_t);\tau)$
      & $\delta_t[x_t]\,\xi_t[x_t]$
      & Normalized sum~\eqref{eq:generalized-kgw-stat} \\
    Adaptive KGW~\cite{liu2024adaptive}
      & D
      & $\{0,1\}^{|\mathcal{V}|}$ (adaptive $\gamma_t$)
      & $\operatorname{softmax}((\mathbf{l}_t{+}\delta_t\xi_t);\tau)$
      & $\delta_t\,\xi_t[x_t]$
      & Normalized sum~\eqref{eq:generalized-kgw-stat} \\
       GumbelSoft~\cite{fu2024gumbelsoft}
      & DF
      & $\mathbb{R}^{|\mathcal{V}|}$ (Gumbel)
      & $\operatorname{softmax}((\mathbf{l}_t{+}\xi_t);\tau)$
      & $\xi_t[x_t]$
      & Scaled mean~\eqref{eq:gumbel-stat} \\
    DipMark~\cite{wu2024dipmark}
      & DF
      & Permutation $\iota_t$
      & Rank-reweighted~\eqref{eq:dipmark-decoder}
      & $\mathbf{1}[x_t \in \text{green region}]$
      & $Z$-score~\eqref{eq:dipmark-stat} \\
    ITS~\cite{kuditipudi2024robust}
      & DF
      & $(\mu_t, \pi_t)$
      & Inverse CDF~\eqref{eq:its-decoder}
      & $|\mu_t - \eta(\pi_t(x_t))|$
      & Negative mean~\eqref{eq:its-stat} \\
    SynthID~\cite{dathathri2024scalable}
      & DF
      & $\mathbb{R}^{|\mathcal{V}|}$ (proprietary)
      & Tournament Sampling~\eqref{eq:synthid-layer-update}
      & 
      $\xi_t[x_t]$
      & Mean~\eqref{eq:synthid-mean-statistic} or Bayesian~\eqref{eq:synthid-bayesian-statistic}\\
    \bottomrule
  \end{tabular}}
\end{table*}

\section{Background}
\label{sec:background}

We introduce the notation used throughout the paper, give a
self-contained description of the generative LLM watermarking
paradigm, and formally specify each of the seven watermarking
schemes evaluated in \S\ref{sec:experiments}.
Table~\ref{tab:scheme-summary} summarises the schemes and
their key components for reference.

\subsection{Language Model Generation}
\label{sec:background:lm}

An autoregressive language model generates a token sequence
$\mathbf{x} = (x_1, x_2, \ldots, x_T)$ over a vocabulary
$\mathcal{V}$ one token at a time.
At decoding step $t$, the model maps the preceding context
$x_{1:t-1}$ to a logit vector
$\mathbf{l}_t \in \mathbb{R}^{|\mathcal{V}|}$, where
$\mathbf{l}_t[x]$ is the unnormalized score of
token $x \in \mathcal{V}$.
The next-token distribution is
\[
  p_t(x \mid x_{1:t-1}) = \operatorname{softmax}(\mathbf{l}_t; \tau)
  = \frac{\exp(\mathbf{l}_t[x] / \tau)}
         {\sum_{x' \in \mathcal{V}} \exp(\mathbf{l}_t[x'] / \tau)},
  \label{eq:softmax}
\]
where $\tau > 0$ is the temperature.
The next token is then sampled as
$x_t \sim p_t(\cdot \mid x_{1:t-1})$.
We write $p_t \in \Delta^{|\mathcal{V}|-1}$ as the simplex
of next-token probabilities at step~$t$,
and drop the conditioning notation and decoding step  when the context is clear.

\subsection{The Generative Watermarking Paradigm}
\label{sec:background:paradigm}

Generative LLM watermarking~\cite{kirchenbauer2023watermark,christ2024undetectable,kuditipudi2024robust}
embeds a hidden
statistical signal into text during token generation, making
the signal imperceptible to human but reliably
detectable by an algorithm.
All schemes evaluated here instantiate the following
two-component framework.

\vspace{+0.05in} \noindent {\bf Watermark generator.}
At each decoding step~$t$, the generator receives the
preceding context $x_{1:t-1}$ and the model's logit vector
$\mathbf{l}_t$, and produces the next token~$x_t$ via:

\begin{enumerate}[leftmargin=*, label=(\roman*), nosep]
  \item \textbf{PRF evaluation.}  A pseudorandom function
        (PRF) $F_{sk}$ keyed by a secret key $sk$ is
        applied to a \emph{context window}
        $c_t \subseteq x_{1:t-1}$ (e.g., the last $h$
        tokens) to produce a pseudorandom vector
        $\xi_t = F_{sk}(c_t) \in \mathcal{E}$, where the
        range $\mathcal{E}$ is scheme-specific.
  \item \textbf{Watermarked sampling.}  A decoder
        $\Gamma(\xi_t, \mathbf{l}_t)$ selects the next token
        $x_t$ using both $\mathbf{l}_t$ and $\xi_t$,
        inducing a watermarked distribution
        $\hat{p}_t \in \Delta^{|\mathcal{V}|-1}$ that may or may not
        differ from the base distribution $p_t$.
\end{enumerate}

\vspace{+0.05in} \noindent {\bf Watermark detector.}
Given a candidate text $\mathbf{x} = (x_1, \ldots, x_T)$,
the detector consists of below steps:
\begin{enumerate}[label=(\roman*), leftmargin=*, nosep]
  \item recomputes $\xi_t = F_{sk}(c_t)$ for each step~$t$
        using the same secret key $sk$;
  \item evaluates a \emph{token score function}
        $s_t = \phi(\xi_t, x_t) \in \mathbb{R}$ that measures the
        per-token watermark evidence; and
  \item aggregates scores into a \emph{detection statistic}
        $\Phi(s_{1:T})$,
        which is compared against a pre-calibrated threshold
        $\theta$ to render a binary decision.
\end{enumerate}
The threshold $\theta$ is set to control the false positive
rate (FPR) under the null hypothesis $H_0$ that the text
was sampled from the unwatermarked base model.
The test statistic $\Phi(s_{1:T})$ is
\emph{designed} to be large in expectation under $H_1$
(watermarked text) and small under~$H_0$, so that
watermarked text is detected with high true positive rate
(TPR) at a fixed FPR.

\vspace{+0.05in} \noindent {\bf Categorization.}
Watermarking schemes are broadly classified as
\emph{distortional} if $\hat{p}_t \neq p_t$ with positive
probability (so that the output distribution differs from
the base model), and \emph{distortion-free} if
$\mathbb{E}[\hat{p}_t] = p_t$ in expectation over the
randomness of $\xi_t$~\cite{christ2024undetectable}.
KGW and its variants, i.e., Adaptive KGW \cite{liu2024adaptive} and Token-Specific KGW \cite{huo2024tokenspecific} are
distortional, while GumbelSoft, DipMark, ITS and SynthID are distortion-free.

\subsection{LLM Watermarking Schemes}

\subsubsection{KGW}
\label{sec:background:kgw}

KGW~\cite{kirchenbauer2023watermark}
partitions the vocabulary into \emph{green} and \emph{red}
tokens randomly, then applies a logit boost to all
green tokens.

\begin{itemize}[leftmargin=*,nosep]
  \item \textbf{Context.} The current token $x_t$ itself
        (self-hash variant; the partition is seeded by the
        token being generated rather than the preceding
        context, so the partition is fixed for each
        vocabulary token across all positions).

  \item \textbf{PRF.} $F_{sk}(x_t)$ maps $x_t$ to a seed,
        which is used to sample a binary vector
        $\xi_t \in \{0,1\}^{|\mathcal{V}|}$ with exactly
        $\gamma|\mathcal{V}|$ ones (the \emph{green list},
        $\mathcal{G}_t$) and $(1-\gamma)|\mathcal{V}|$
        zeros (the \emph{red list}).
        Here $\gamma \in (0,1)$ is the green-list fraction.

  \item \textbf{Decoder.}
        $\Gamma(\xi_t, \mathbf{l}_t)$ samples from
        $\operatorname{softmax}(\frac{\mathbf{l}_t + \delta \cdot \xi_t}{\tau})$,
        where $\delta > 0$ is the logit boost.

  \item \textbf{Token score.}
        $s_t = \phi(\xi_t, x_t) = \xi_t[x_t] \in \{0,1\}$,
        i.e., 1 if $x_t$ is green and 0 otherwise.

  \item \textbf{Detection statistic.} KGW uses the Z-score defined below:
       \begin{align}
          \Phi(s_{1:T})
          = \frac{\sum_{t=1}^{T} s_t - \gamma T}
                 {\sqrt{T\gamma(1-\gamma)}}.
          \label{eq:kgw-stat}
        \end{align}
        Under $H_0$, $\Phi \xrightarrow{d} \mathcal{N}(0,1)$;
        under $H_1$, the elevated green-token rate shifts
        $\Phi$ positively.
\end{itemize}

\subsubsection{KGW Variants}
\label{sec:background:generalized-kgw}

Several extensions of
KGW replace its fixed green-list
fraction and logit boost with context-dependent parameters.
Two representative variants are \emph{Token-Specific
KGW}~\cite{huo2024tokenspecific}, which assigns different
parameters to individual vocabulary tokens, and
\emph{Adaptive KGW}~\cite{liu2024adaptive}, which dynamically adjusts
shared parameters  at each decoding step.
\emph{We can unify both schemes using the following generalized
KGW formulation.}

\begin{itemize}[leftmargin=*,nosep]
  \item \textbf{Context.}
        The previous $h$ tokens $x_{t-h:t-1}$, optionally
        together with semantic-coherence, complexity, or
        token-importance signals derived from current
        decoding context.

  \item \textbf{PRF and green-list construction.}
        The pseudorandom function
        $F_{sk}(x_{t-h:t-1})$ seeds a context-dependent
        binary vector $\xi_t$, where
        $\xi_t[x] \sim
          \operatorname{Bernoulli}\!\left(\gamma_t[x]\right),
          \, \forall\,x\in\mathcal{V},$

        and $\gamma_t[x]\in(0,1)$ is the probability that
        token $x$ is assigned to the green list at step $t$ and each token $x$ has a boost $\delta_t[x]$.

        \begin{enumerate}
            \item {\bf Token-Specific KGW:} Both $\gamma_t[x]$ and the
        boost $\delta_t[x]$ may vary across tokens according
        to semantic similarity or context-aware importance. I.e., $\gamma_t[x]$
        \item {\bf Adaptive KGW:} All tokens at a given step $t$ typically
        share the same parameters, i.e., $\gamma_t[x]=\gamma_t, \delta_t[x]=\delta_t, \forall\,x\in\mathcal{V}$; but

        $\gamma_t$ and $\delta_t$ vary across decoding
        steps according to semantic coherence, sentence
        complexity, topic sensitivity, or token importance.
       Thus, Adaptive KGW is a step-dependent special case of the generalized token-dependent formulation.
        \end{enumerate}

  \item \textbf{Decoder.}
        Let $\boldsymbol{\delta}_t \geq 0$ denote the boost vector.

        For Token-Specific KGW, the watermarked
        token is sampled via
        $
          x_t \sim
          \operatorname{softmax}\!\left(
            \frac{\mathbf{l}_t+
            \boldsymbol{\delta}_t\odot\boldsymbol{\xi}_t}
            {\tau}
          \right),
        $
        where $\odot$ is element-wise product.

        For Adaptive KGW, this becomes
        $
          x_t \sim
          \operatorname{softmax}\!\left(
            \frac{\mathbf{l}_t+\delta_t \cdot \boldsymbol{\xi}_t}
            {\tau}
          \right).
       $

  \item \textbf{Token score.}
        In Token-Specific KGW,
        the score contributed by
        the selected token is $s_t=\phi(\boldsymbol{\xi}_t,x_t)=\delta_t[x_t]\,\xi_t[x_t]$; while
        under Adaptive KGW, the
        score becomes
        $s_t=\delta_t\,\xi_t[x_t].$

  \item \textbf{Detection statistic.}
        Both variants aggregate the token-level scores using
        the normalized sum  statistic as below
        \begin{align}
          \Phi(s_{1:T})
          =
          \frac{
            \sum_{t=1}^{T}s_t
            -
            \mathbb{E}_{H_0}\!\left[
              \sum_{t=1}^{T}s_t
            \right]
          }{
            \sqrt{
              \operatorname{Var}_{H_0}\!\left[
                \sum_{t=1}^{T}s_t
              \right]
            }
          },
          \label{eq:generalized-kgw-stat}
        \end{align}
        where the null-hypothesis mean and variance are
        computed using the corresponding token-specific or
        step-adaptive green-list parameters. A sufficiently
        large value of $\Phi(s_{1:T})$ indicates the presence
        of the watermark.
\end{itemize}

\subsubsection{GumbelSoft}
\label{sec:background:gumbelsoft}

GumbelSoft~\cite{fu2024gumbelsoft} injects Gumbel-distributed
pseudorandom noise \cite{gumbel1958statistics} into the logits before sampling.
Because the Gumbel-Max trick is equivalent to sampling from
the softmax distribution, the scheme is {distortional-free}.

\begin{itemize}[leftmargin=*,nosep]
  \item \textbf{Context.} The previous $h$ tokens
        $x_{t-h:t-1}$.

  \item \textbf{PRF.} $F_{sk}(x_{t-h:t-1})$ hashes the
        context to a seed and draws
        $
          \xi_t[x] \overset{\mathrm{i.i.d.}}{\sim}
          \operatorname{Gumbel}(0,1),
          \quad \forall\, x \in \mathcal{V}.
          \label{eq:gumbel-prf}
        $

  \item \textbf{Decoder.}
        $x_t \sim \operatorname{softmax}((\mathbf{l}_t + \xi_t)/\tau)$.

  \item \textbf{Token score.}
        $s_t = \phi(\xi_t, x_t) = \xi_t[x_t]$,
        i.e., the Gumbel noise value at the sampled
        position.

  \item \textbf{Detection statistic.} It uses scaled mean defined as:
        \begin{align}
          \Phi(s_{1:T})
          = \frac{\sqrt{6T}}{\pi}
            \left(\frac{1}{T}\sum_{t=1}^{T} s_t - \gamma_{\mathrm{EM}}\right),
          \label{eq:gumbel-stat}
        \end{align}
        where $\gamma_{\mathrm{EM}} \approx 0.5772$ is the
        Euler--Mascheroni constant, i.e., the mean of a
        $\operatorname{Gumbel}(0,1)$ random variable.
\end{itemize}

\subsubsection{DipMark}
\label{sec:background:dipmark}

DipMark~\cite{wu2024dipmark} is a distortion-free scheme
that reweights the token distribution via a pseudorandom
permutation, assigning binary colour scores to tokens by
their permuted rank.

\begin{itemize}[leftmargin=*,nosep]
  \item \textbf{Context.} The previous $h$ tokens
        $x_{t-h:t-1}$.

  \item \textbf{PRF.} $F_{sk}(x_{t-h:t-1})$ seeds a
        uniform random permutation
        $\iota_t : \mathcal{V} \to \mathcal{V}$.

  \item \textbf{Decoder.}  Token $\iota_t[i]$ is sampled
        with probability $\lambda_t(i) - \lambda_t(i-1)$,
        where $\lambda_t(i)$ is defined as
        \begin{equation}
          \max\!\left\{\sum_{j=1}^{i} p_t(\iota_t[j])
                         - \alpha,\; 0\right\}
            + \max\!\left\{\sum_{j=1}^{i} p_t(\iota_t[j])
                         - (1-\alpha),\; 0\right\},
        \label{eq:dipmark-decoder}
        \end{equation}
        and $\alpha \in (0, \tfrac{1}{2})$ is a
        scheme parameter.
        Tokens ranked in the upper $\alpha$-fraction of
        the permuted vocabulary (the ``green'' region) are
        upweighted, while tokens in the lower $\alpha$-fraction
        (the ``red'' region) are downweighted.

  \item \textbf{Token score.}
        $s_t = \phi(\xi_t, x_t)
        = \mathbf{1}\{x_t \in \iota_t[\lceil(1-\gamma)|\mathcal{V}|\rceil :|\mathcal{V}|]\}$,
        i.e., 1 if the sampled token falls in the green
        (high-rank) portion of the permuted vocabulary,
        0 otherwise.

  \item \textbf{Detection statistic.} It uses the Z-score defined as:
        \begin{align}
          \Phi(s_{1:T})
          = \frac{\sum_{t=1}^{T} s_t - (1-\gamma)T}
                 {\sqrt{T}}.
          \label{eq:dipmark-stat}
       \end{align}
\end{itemize}

\subsubsection{ITS}
\label{sec:background:its}

The Inverse Transform Sampling (ITS) scheme~\cite{kuditipudi2024robust}
is distortion-free by construction: it uses a pre-committed
global key sequence to deterministically select tokens via
inverse CDF sampling, producing an output distribution
equal in expectation to the base model.

\begin{itemize}[leftmargin=*,nosep]
  \item \textbf{Context.} A global watermark key
        $\xi\text{-list}$ and position index $t$.
        Each entry $\xi\text{-list}[t] = (\mu_t, \pi_t)$
        consists of a random threshold $\mu_t \in \text{Uni}(0,1)$
        and a permutation $\pi_t$ of $\mathcal{V}$.

  \item \textbf{PRF.} The PRF simply reads off the
        pre-committed key: $F_{sk}(t) = (\mu_t, \pi_t)$.

  \item \textbf{Decoder.} For a fixed key $(\mu_t,\pi_t)$, the decoder selects
        \begin{align}
        \small
          x_t = \pi_t^{-1}\!\left(
            \min\!\left\{
              \pi_t(i) :
              p_t\bigl(\{j : \pi_t(j) \leq \pi_t(i)\}\bigr)
              \geq \mu_t
            \right\}
          \right),
          \label{eq:its-decoder}
        \end{align}
      which selects the smallest $\pi_t$-ranked token whose cumulative probability under the permuted ordering meets or exceeds $\mu_t$. This is equivalent to inverse-CDF sampling over the permuted token distribution.

  \item \textbf{Token score.}
        \begin{align}
        \small
          s_t = \phi(\xi_t, x_t)
          = \left|\mu_t - \eta(\pi_t(x_t))\right|,
          \,\,
          \eta(k) = \frac{k-1}{|\mathcal{V}|-1},
          \label{eq:its-score}
       \end{align}
        which measures how closely the sampled token's
        normalized rank matches the random threshold.
        Under $H_1$, scores are concentrated near zero;
        under $H_0$, they are approximately uniform on
        $[0, \tfrac{1}{2}]$.

  \item \textbf{Detection statistic.} It uses a negative mean:
        \begin{align}
          \Phi(s_{1:T})
          = -\frac{1}{T}\sum_{t=1}^{T} s_t.
          \label{eq:its-stat}
        \end{align}
        Note that larger negative values of $\Phi$ indicate
        stronger watermark evidence.
\end{itemize}

\subsubsection{SynthID}
\label{sec:background:synthid}

SynthID~\cite{dathathri2024scalable} is the first production-ready generative LLM watermarking system developed by Google and achieves the state-of-the-art detection performance.
Its core contribution is
a novel \emph{Tournament Sampling} mechanism that 
samples several candidate tokens from base distribution
and repeatedly selects candidates having larger pseudorandom
$g$-values.

\begin{itemize}[leftmargin=*,nosep]
    \item \textbf{Context.}
    The previous $h$ tokens
        $x_{t-h:t-1}$.
    \item \textbf{PRF and $g$-value function.}
    The random seed at position $t$ is computed from the
    preceding context and the secret watermark key $r_t
        =
        F_{sk}(x_{t-h:t-1}).$

    For a tournament with $m$ layers, the seed $r_t$ determines
    $m$ independent pseudorandom functions (also called $g$-value functions):
    \begin{align}
    \label{eq:gvalue}
        g_{t,\ell}(i)
        :=
        g_\ell(i,r_t),
        \quad
        \ell\in\{1,\ldots,m\},\quad i\in\mathcal V.
    \end{align}
    Each function assigns every token an ordered
    $g$-value in a set $\mathcal Z$.
    The standard construction
    uses binary values, but it naturally
    extends to arbitrary discrete ordered $g$-values.

    \item \textbf{Decoder: Tournament Sampling.}
    Assume the tournament match contains $N\geq2$
    competitors. The decoder initially samples $N^m$
    candidate tokens independently from  base distribution $p_t$. At layer
    $\ell$, the remaining candidates are divided into groups
    of $N$, and the candidate having the largest
    $g_{t,\ell}$-value in each group is retained. Ties are
    broken uniformly. After $m$ layers, the only remaining
    token becomes $x_t$.

    Equivalently, let $ p_t^{(0)}:=p_t$
    and $p_t^{(\ell)}$ denote the winner distribution after
    layer $\ell$. For a score level $a\in\mathcal Z$, define the cumulative masses under the
distribution $p^{(\ell-1)}$ by
    \begin{align}
        & Q_{t,\ell,<}(a)
        :=
        \sum_{j:g_{t,\ell}(j)<a}
        p_{t,j}^{(\ell-1)},
        Q_{t,\ell,\leq}(a),
        :=
        \sum_{j:g_{t,\ell}(j)\leq a}
        p_{t,j}^{(\ell-1)},
Q_{t,\ell,=}(a)
        :=
        Q_{t,\ell,\leq}(a)-Q_{t,\ell,<}(a).
    \label{eq:cumdist}
    \end{align}
    The layerwise Tournament-sampling  update is
    \begin{equation}
        p_{t,i}^{(\ell)}
        =
        p_{t,i}^{(\ell-1)}
        \frac{
            Q_{t,\ell,\leq}(g_{t,\ell}(i))^N
            -
            Q_{t,\ell,<}(g_{t,\ell}(i))^N
        }{
            Q_{t,\ell,=}(g_{t,\ell}(i))
        }.
        \label{eq:synthid-layer-update}
    \end{equation}

     The final watermarked distribution is
    $
        \widehat p_t:=p_t^{(m)},
        \,
        x_t\sim\widehat p_t.
    $

    \item \textbf{Detection statistics.}
    SynthID supports both
    the mean score and
    Bayesian score, and larger values provide stronger watermark evidence.

    \begin{itemize}[leftmargin=*]
        \item \textbf{Mean Score.}
        The scalar contribution of token $x_t$ is $s_t^{\mathrm{MS}}
            :=
            \frac{1}{m}
            \sum_{\ell=1}^{m}g_{t,\ell}(x_t).$

        Let $\mathcal I$ be the set of positions retained after
        context-repetition and other validity masking. The
        document-level mean statistic is
        \begin{equation}
        \small
            \Phi_{\mathrm{MS}}(x_{1:T})
            :=
            \frac{1}{m|\mathcal I|}
            \sum_{t\in\mathcal I}
            \sum_{\ell=1}^{m}
            g_{t,\ell}(x_t).
            \label{eq:synthid-mean-statistic}
        \end{equation}

        Detection is performed by comparing
        $\Phi_{\mathrm{MS}}$ with a threshold calibrated under
        the unwatermarked null distribution.

        \item \textbf{Bayesian Score.}
        Let $\psi_{t,\ell}\in\{1,\ldots,N\}$ denote the
        number of distinct competitors
        associated with the observed token at layer $\ell$. 
        For a fixed value $\psi_{t,\ell}=c$, the likelihood
        ratio associated with observing $g$-value $z$ is
        \begin{align}
            \Lambda_c(z)
            :=
            \frac{
                \Pr(g_{t,\ell}=z
                \mid H_1,\psi_{t,\ell}=c)
            }{
                \Pr(g_{t,\ell}=z\mid H_0)
            }
            \label{eq:synthid-fixed-c-bayes-factor}
        \end{align}

        Given prior layer values
        $
            \mathbf g_{t,<\ell}
            =
            \bigl(
                g_{t,1}(x_t),\ldots,
                g_{t,\ell-1}(x_t)
            \bigr),
        $
        the conditional Bayes factor for layer $\ell$ is
        \begin{equation}
            \textrm{BS}_{t,\ell}
            \bigl(
                g_{t,\ell}(x_t);
                \mathbf g_{t,<\ell}
            \bigr)
            :=
            \sum_{c=1}^{N}
            \Pr\!\left(
                \psi_{t,\ell}=c
                \mid
                \mathbf g_{t,<\ell},H_1
            \right)
            \Lambda_c\!\left(g_{t,\ell}(x_t)\right).
            \label{eq:synthid-layerwise-bayes-factor}
        \end{equation}

        The log-Bayesian contribution of token $x_t$ is
        \begin{equation}
        \small
            s_t^{\mathrm{BS}}
            :=
            \sum_{\ell=1}^{m}
            \log
            \textrm{BS}_{t,\ell}
            \bigl(
                g_{t,\ell}(x_t);
                \mathbf g_{t,<\ell}
            \bigr).
            \label{eq:synthid-token-bayesian-score}
        \end{equation}
        The document-level log-Bayes statistic is therefore
        \begin{equation}
        \small
            \Phi_{\mathrm{BS}}(x_{1:T})
            :=
            \sum_{t\in\mathcal I}s_t^{\mathrm{BS}}
            =
            \sum_{t\in\mathcal I}
            \sum_{\ell=1}^{m}
            \log
            \textrm{BS}_{t,\ell}
            \bigl(
                g_{t,\ell}(x_t);
                \mathbf g_{t,<\ell}
            \bigr).
            \label{eq:synthid-bayesian-statistic}
        \end{equation}

            \end{itemize}
      \end{itemize}

%% file: AttackFramework.tex
\section{Attack Framework}
\label{sec:attack}

We consider a watermark applied at generation time of an LLM by modifying its token
selection distribution. 
The watermark is considered secure if it is simultaneously:
(i) statistically detectable with high confidence, and 
(ii) minimally distortive with respect to the raw model distribution. 
Our attack challenges 
property (i) while preserving property (ii)\footnote{For brevity, in the rest of the paper, we omit the decoding step index $t$ and write $s_j$ instead of $s_{t}(j)$ to denote the $j$-th token at decoding step $t$}.

\subsection{Threat Model}
\label{sec:threatmodel}

\vspace{+0.05in}
\noindent {\bf Adversary capabilities:} 
We assume a \emph{distribution-level black-box adversary} with below capabilities:

\begin{itemize}[leftmargin=*,nosep]
    \item \textbf{Sampling and detector access:}
    The adversary knows the watermark detection rule
    (e.g., score aggregation and threshold), since it knows which watermark to attack; the attacker and can query the watermarked LLM
    to observe generated text samples.

    \item \textbf{No seed or key access:}
    The adversary has no access to the permutation, random seed,
    or secret partition used to construct the watermark.

    \item \textbf{No model access:}
    The adversary does not modify the model weights or have 
access to the watermark algorithm.
\end{itemize}

\vspace{+0.05in}
\noindent {\bf Adversary objective:}
The adversary objective is to produce an attacked distribution $\tilde p$ such that:

\begin{enumerate}[leftmargin=*,nosep]
    \item \textbf{Watermark evasion:}
    Generated text 
    $x$ sampled 
    from $\tilde p$ is classified as \emph{non-watermarked} by the
    detector with high probability. 
    
    \item \textbf{Low distortion:}
    The modified distribution $\tilde p$ remains close to the original model distribution $p$.

\end{enumerate}

In other words, the adversary aims to \emph{erase the watermark signal while preserving
language quality}.

\subsection{Attack Principle}

The central quantity we study is the \emph{expected
per-token score} $\mathbb{E}[s_j]$ as a function of
the token index $j$. 
In \S\ref{sec:optimality} 
we show that
for every watermarking  scheme in Table~\ref{tab:scheme-summary}, this
function is \emph{monotone} in the probability ratio
$\hat{p}_j / p_j$, a property we call
\emph{score-ratio alignment} and forms the theoretical
foundation of our attack.

Specifically, the attack exploits the fact that the schemes (both distortional and distortion-free) induce a structured change
in token probabilities that correlates with the watermark score $s_j$, i.e., 
tokens with the smallest expected scores experience the
largest score-ratio degradation:
\begin{equation}
\label{eq:scoreratio}
    j^\star =\arg\min_{j \in \mathcal{V}} \frac{\hat p_j}{p_j}
    \quad \Rightarrow \quad
    j^\star \in \arg\min_{j \in \mathcal{V}} \mathbb{E}[s_j].
\end{equation}

The adversary leverages this alignment to \emph{estimate which tokens contribute most 
to the detection statistic} and counteracts the watermark by applying a compensating
reweighting or noise injection that neutralizes the expected score bias.

\begin{algorithm}[!t]
\small
\caption{Black-Box Attack on LLM Watermarking}
\label{alg:blackbox_watermark_attack}
\begin{algorithmic}[1]

\REQUIRE Prompt $x$, surrogate model $M_s$, watermarked oracle $\widetilde{M}$,
candidate size $K$, query budget $N$, and generation length $T$
\ENSURE Attacked continuation $y$

\STATE Initialize $y \leftarrow [\,]$

\FOR{$j = 1,\ldots,T$}

    \STATE Set the current context $h \leftarrow (x,y)$

    \STATE Query the watermarked oracle $N$ times:
    \[
        \widetilde{t}^{(1)},\ldots,\widetilde{t}^{(N)}
        \sim \widetilde{M}(\cdot \mid h)
    \]

    \STATE Select the top-K candidate tokens using the surrogate model:
    \[
        \mathcal{C}
        \leftarrow
        \operatorname{TopK}\!\left(p_{M_s}(\cdot \mid h), K\right)
    \]

    \STATE Estimate the watermarked probability for each $t \in \mathcal{C}$:
    \[
        \widehat{p}_{\mathrm{wm}}(t \mid h)
        \leftarrow
        \frac{
            \sum_{i=1}^{N}\mathbb{I}
            [\widetilde{t}^{(i)} = t] + \epsilon
        }{
            \sum_{u \in \mathcal{C}}
            \sum_{i=1}^{N}\mathbb{I}
            [\widetilde{t}^{(i)} = u]
            + \epsilon |\mathcal{C}|
        },
        \quad \forall t \in \mathcal{C}
    \]

    \STATE Select the next token:
    \[
        t_j^\star
        \leftarrow
        \arg\min_{t \in \mathcal{C}}
        \frac{
            \widehat{p}_{\mathrm{wm}}(t \mid h)
        }{
            p_{M_s}(t \mid h) + \epsilon
        }
    \]

    \STATE Append $t_j^\star$ to the continuation:
    $
        y \leftarrow y \circ t_j^\star
    $

    \IF{$t_j^\star = \mathrm{EOS}$}
        \STATE \textbf{break}
    \ENDIF

\ENDFOR

\RETURN $y$

\end{algorithmic}
\end{algorithm}

\subsection{Attack Algorithm}

Algorithm~\ref{alg:blackbox_watermark_attack} presents our detailed attack procedure.
Here, we choose an arbitrary surrogate model and use its tokenizer consistently throughout the attack.
We make the following remarks regarding Algorithm~\ref{alg:blackbox_watermark_attack}.

\begin{enumerate}[leftmargin=*]
    \item \emph{Computational consideration.}
    Unlike our theoretical analysis that considers the entire vocabulary $\mathcal{V}$, in practice, we estimate the watermarked probabilities only for the top-$K$ candidate tokens in $\mathcal{C}$ selected by the surrogate model to reduce computational cost. Most tokens outside $\mathcal{C}$ have very small probabilities and are thus unlikely to be selected. This provides an efficiency--search-space tradeoff by substantially reducing the number of candidate tokens while retaining the most likely ones.

    \item \emph{Per-token queries.}
    Our attack operates at the token level. When querying the watermarked oracle, we set \texttt{max\_new\_tokens=1} to obtain a single next-token sample. If generation-length control is unavailable and the oracle returns a sequence, we use only its first generated token as the next-token observation.

    \item \emph{Numerical stability.}
    We use a small constant $\epsilon>0$ in the denominator to avoid numerical instability. Since $\epsilon$ is sufficiently small, it has a negligible effect on the probability ratios and does not affect the attack results in practice.
\end{enumerate}

%% file: Optimality.tex
\section{Optimality: Score-Ratio Minimization}
\label{sec:optimality}

We now formally establish that the score-ratio minimization strategy in Equation (\ref{eq:scoreratio}) is optimal for the studied watermark schemes. 
Specifically, we prove that in every scheme, the token achieving
the minimum probability ratio $\hat{p}_j / p_j$ also achieves the minimum
expected detection score $\mathbb{E}[s_j]$---the score-ratio alignment
condition of Equation  \eqref{eq:scoreratio}. This alignment
confirms that an adversary executing Algorithm~\ref{alg:blackbox_watermark_attack} could minimize 
the detection statistic in expectation, without requiring knowledge of the secret key or the watermark seed.

Next, we describe the setup and present the theorems. 
Since both KGW and GumbelSoft belong to the class of logit-addition watermarking schemes, the theorem applies directly to both methods. Furthermore, the two KGW variants can be analyzed under the same unified framework.
{\bf All proofs are deferred to Appendix \ref{app:proof-opt}}.

\subsection{Logit-Addition Watermark: KGW and GumbelSoft}
\label{subsec:proof-logit}

\noindent\textbf{Setup.}
Let $l \in \mathbb{R}^{|\mathcal{V}|}$ denote the original logit vector
produced by the base LLM at a given decoding step, and define the
base-model next-token distribution as
$
    p_j = \mathrm{softmax}(l)_j
        = \frac{e^{l_j}}{\sum_{k \in \mathcal{V}} e^{l_k}}.
$
A logit-addition watermark injects a per-token signal $\xi_j$ into the
logits before sampling:
\begin{align}
    \hat{l}_j = l_j + \delta\,\xi_j,
    \qquad
    \hat{p}_j = \mathrm{softmax}(\hat{l})_j,
    \qquad
    s_j := \xi_j,
\end{align}
where $\delta > 0$ is the watermark strength and $\xi_j$ is drawn from a
distribution determined by the PRF and the current context. For KGW,
$\xi_j \sim \mathrm{Bernoulli}(\gamma)$; for GumbelSoft,
$\xi_j \sim \mathrm{Gumbel}(0,1)$. In both cases the detector score for
the sampled token is $s_j = \xi_j$.

\begin{theorem}[Score-Ratio Alignment: Logit-Addition]
\label{thm:logit-addition}
For any logit-addition watermark with $\delta > 0$, we have:
\[
    \arg\min_{j \in \mathcal{V}} \frac{\hat{p}_j}{p_j}
    \;\subseteq\;
    \arg\min_{j \in \mathcal{V}} \mathbb{E}[s_j].
\]
\end{theorem}

\begin{remark}
\label{rem:logit-intuition}
Equation  \eqref{eq:ratio-logit} shows that the probability ratio
$\hat{p}_j / p_j$ is a monotone function of $\xi_j$ alone: the original
logit $l_j$ cancels exactly. Consequently, any token selection strategy
that minimizes $\xi_j$---equivalently, minimizes the detection score---also
minimizes the probability ratio, and vice versa. This symmetry is the
algebraic source of score-ratio alignment in all logit-addition schemes.
\end{remark}

\subsection{{KGW Variants} }
\label{subsec:proof-kgw-variant}

\noindent\textbf{Setup.}
In   KGW variants~\cite{huo2024tokenspecific,liu2024adaptive}, each token
$j \in \mathcal{V}$ receives an independent binary color
$
\xi_j \sim \mathrm{Bernoulli}(\gamma_j),
$
and a boost magnitude $\delta_j \geq 0$.  Watermarked
logit and probability are
\begin{align}
    \hat{\ell}_j = \ell_j + \delta_j\,\xi_j,
    \qquad
    \hat{p}_j = \mathrm{softmax}(\hat{\ell})_j,
    \qquad
    s_j = \xi_j.
\end{align}
The probability ratio simplifies to
\begin{align}
\small
    \frac{\hat{p}_j}{p_j}
    = \frac{e^{\ell_j + \delta_j\xi_j} / \sum_j e^{\ell_j + \delta_j\xi_j}}
            {e^{\ell_j}                  / \sum_j e^{\ell_j}}
     = e^{\delta_j\xi_j}
       \cdot \underbrace{\frac{\sum_j e^{\ell_j}}
                              {\sum_j e^{\ell_j + \delta_j\xi_j}}}_{=:\,C},
     \label{eq:ratio-kgw-variant}
\end{align}
where $C$ is a constant,
so
$\arg\min_i \hat{p}_j / p_j = \arg\min_j \delta_j\xi_j$.

\begin{theorem}[Score-Ratio Alignment: KGW Variants]
\label{thm:kgw-variant}
For every realization of the token-specific green list, we have:
\[
    \arg\min_{j \in \mathcal{V}} \frac{\hat{p}_j}{p_j}
    \;\subseteq\;
    \arg\min_{j \in \mathcal{V}} s_j.
\]
\end{theorem}

\begin{theorem}[Feasibility of Red-Token Minimizers: Token-Specific KGW]
\label{thm:kgw-feasibility}
In the Token-Specific KGW watermark, every red token is a feasible
minimizer of the probability ratio. In particular, if at least one token
has $\xi_j = 0$, then
\[
    \arg\min_{j \in \mathcal{V}} \frac{\hat{p}_j}{p_j}
    = \{j : \xi_j = 0\}
    = \arg\min_{j \in \mathcal{V}} s_j.
\]
If no red token exists, every token is a feasible minimizer.
\end{theorem}

\subsection{DipMark}
\label{subsec:proof-dipmark}

\noindent\textbf{Setup.} 
In DipMark, a
step samples a uniformly random permutation
$
    \iota : \{0, 1, \ldots, |\mathcal{V}|-1\} \to \mathcal{V},
$
which induces a shuffled ordering of the vocabulary $\mathcal{V}$. The binary watermark
color (score) assigned to token $\iota[k]$ is
\begin{align}
    s_{\iota[k]} =
    \begin{cases}
        0 & \text{(red)}   \quad \text{if } k < \lceil \gamma|\mathcal{V}| \rceil, \\
        1 & \text{(green)} \quad \text{otherwise,}
    \end{cases}
\end{align}
which means the first $\lceil \gamma|\mathcal{V}| \rceil$ tokens in the permuted order form the red prefix.

We use
$\hat{p}_j / p_j$ as the \emph{degradation factor} of token $j$.

\begin{theorem}[Score-Ratio Alignment: DipMark]
\label{thm:dipmark}
In DipMark watermarking, the following inclusion holds:
\[
    \arg\min_{j \in \mathcal{V}} \frac{\hat{p}_j}{p_j}
    \;\subseteq\;
    \arg\min_{j \in \mathcal{V}} \mathbb{E}[s_j].
\]
Equivalently, every token achieving the smallest probability-ratio also achieves the
smallest expected watermark score.
\end{theorem}

\begin{theorem}[Feasibility of Index 0: DipMark]
\label{thm:dipmark-feasibility}
In DipMark watermarking, $\iota[0]$ is always a feasible minimizer of
$s_j$ in the index space induced by the permutation. In particular,
\[
    \arg\min_{j \in \mathcal{V}} \frac{\hat{p}_j}{p_j}
    = \iota[0]
    \in \arg\min_{j \in \mathcal{V}} \mathbb{E}[s_j].
\]
\end{theorem}

\subsection{ITS}
\label{subsec:proof-its}

\noindent\textbf{Setup.}
Recall from \S\ref{sec:background:its} that the key
at position $t$ is $\xi_t=(\mu_t,\pi_t)$,
where $\mu_t\sim\operatorname{Unif}(0,1)$ and $\pi_t$ is a
random permutation of $\mathcal V$. These keys are
obtained from the pre-committed sequence $\xi\text{-list}$ by
applying a random document-level shift. The following result
uses only the single-position marginal distribution of
$\xi_t$; and independence between different positions is not
required. 
Fix a decoding step $t$ and its preceding context. Let
$p_t$ be the corresponding base-model distribution, define
$    \mathcal S_t:=\{j\in\mathcal V:p_t(j)>0\},
$
and let $x_t=\Gamma(p_t,\xi_t)$
denote the token generated by the ITS decoder in
Equation~\eqref{eq:its-decoder}.
For a permutation $\pi_t$, token $j$ occupies the
permuted-CDF interval
\begin{align}
\small
 I_{t,j}^{\pi_t}
 :=
 \Big[
  \sum_{i:\pi_t(i)<\pi_t(j)}p_t(i),
  \sum_{i:\pi_t(i)\leq\pi_t(j)}p_t(i)
 \Big).
 \label{eq:its-proof-cdf-interval}
\end{align}
Consequently,
$
    x_t=j
    \Longleftrightarrow
    \mu_t\in I_{t,j}^{\pi_t}.
$

\begin{theorem}[ITS detection-statistic reduction]
\label{thm:its}

If $|\mathcal S_t|\geq2$,
\begin{align}
 \mathbb{E}_{\xi_t}
 \left[
  -s(x_t^\star)
 \right]
 <
 \mathbb{E}_{\xi_t}
 \left[
  -s(x_t)
 \right].
 \label{eq:its-proof-statistic-reduction}
\end{align}
Thus, empirical probability-ratio minimization reduces the
expected detection statistic to its key-independent null
baseline.
\end{theorem}

\begin{remark}
The strict inequality follows from the independence of the
attacked token and the detector key. In fact, every token
selection rule based only on auxiliary keys has the same
expected null score $\frac{2|\mathcal{V}|-1}{6(|\mathcal{V}|-1)}$ (see more details in Appendix \ref{subsec:proof-its}). Empirical probability-ratio
minimization determines which underrepresented token is
selected, while key independence is what reduces the expected
ITS detection statistic.
\end{remark}

\subsection{SynthID}
\label{subsec:synthid-alignment}

\paragraph{Setup.}
Consider SynthID Tournament sampling with $N\geq2$
competitors per match and $m\geq1$ tournament layers. Fix a
decoding step, its preceding context, and the corresponding
pseudorandom seed $r$. Let
$
    p^{(0)}:=p
$
be the base-model next-token distribution
and let
$
    \mathcal S:=\{i\in\mathcal V:p_i>0\}
$
denote its support. We restrict all probability ratios to
$\mathcal S$, since $\widehat p_i/p_i$ is undefined when
$p_i=0$.
At layer $\ell\in\{1,\ldots,m\}$, token $i$ is assigned the
possibly real-valued watermark score
$
    g_{\ell,i}:=g_\ell(i,r).
$
Let $p^{(\ell)}$ denote the winner distribution after layer
$\ell$. The final watermarked distribution is therefore
$
    \widehat p:=p^{(m)}.
$

\begin{theorem}[Score-Ratio Alignment: SynthID, Mean Score]
\label{thm:synthid-layerwise-alignment}
Consider Tournament sampling with arbitrary $N\geq2$ and
arbitrary discrete, ordered $g$-values. For every layer
$\ell$ and every $i,j\in\mathcal S$,
$ g_{\ell,i}<g_{\ell,j}
 \Longrightarrow
 \hat{p}^{(\ell)}_i/p^{(\ell)}_i< \hat{p}^{(\ell)}_j/p^{(\ell)}_j.
 \label{eq:synthid-strict-monotonicity}
$
Consequently,
\begin{align}
 \arg\min_{i\in\mathcal S} \frac{\hat{p}^{(\ell)}_i}{p^{(\ell)}_i}
 =
 \arg\min_{i\in\mathcal S}g_{\ell,i}.
 \label{eq:synthid-layer-minimizer}
\end{align}
\end{theorem}

\begin{theorem}[Score-Ratio Alignment: SynthID, Bayesian Score]
\label{thm:synthid-bayesian-layerwise}
At any layer $\ell$ and condition on preceding $g$-values,
$
 \arg\min_{i\in\mathcal S}
 \frac{\widehat p^{(\ell)}_i}{p^{(\ell)}_i}
 =
 \arg\min_{i\in\mathcal S}g_{\ell,i}
 \subseteq
 \arg\min_{i\in\mathcal S}s(i).
$$ \label{eq:synthid-single-layer-bayesian-alignment}
$
\end{theorem}

{
\begin{remark}
For fixed preceding $g$-values, the layerwise probability
ratio, 
the mean-score contribution
$g_{\ell,i}$, and the conditional Bayes factor are
nondecreasing in $g_{\ell,i}$. Thus, minimum-ratio selection
minimizes both the mean and Bayesian evidence contributed by
that layer.
\end{remark}}

%% file: Tradeoff.tex
\section{Detectability--Robustness Tradeoff}
\label{sec:tradeoff}

\S\ref{sec:optimality} establishes that score-ratio alignment holds
across all watermarks considered in this work.
A natural question
arises: \emph{Is this alignment an artifact of particular design choices, or reflecting a deeper structural constraint inherent to the generator--detector
paradigm?} In this section, we
show that
score-ratio alignment is
the \emph{inevitable
consequence} of designing a statistically optimal detector. This yields a
fundamental impossibility: no watermarking scheme within the current
paradigm can simultaneously maximize detection power and resist score-ratio minimization adversaries.
{\bf Proofs of the  theorems are in Appendix \ref{app:proof-tradeoff}}.

\subsection{The Fundamental Tradeoff}
\label{subsec:fundamental-tradeoff}

We first formalize the notion of a watermark that can be evaded by our
attack framework.

\begin{definition}[Defeatable Watermark]
\label{def:defeatable}
A watermarking scheme with per-token score function $s_j$ is called
\emph{defeatable} if the token minimizing the probability ratio coincides
with the token minimizing the expected detection score: $ \arg\min_{i \in \mathcal{V}} \frac{\hat{p}_j}{p_j}
    \;=\;
    \arg\min_{j \in \mathcal{V}} \mathbb{E}[s_j].$
\end{definition}

\begin{theorem}[Defeatable Implies Not Robust]
\label{thm:defeatable}
If a watermarking scheme is defeatable, then it is not robust against a
score-ratio-minimizing adversary.
\end{theorem}

\begin{theorem}[Detectability--Robustness Tradeoff]
\label{thm:tradeoff}
Among all hypothesis tests $\mathcal{D}$ for $H_0 : x \sim p$ versus
$H_1 : x \sim \hat{p}$, the statistically optimal detector is the
likelihood-ratio test: $\Lambda(x) = \frac{\hat{p}(x)}{p(x)}.$
Detectability is maximized when the detection score is aligned with the
log-likelihood ratio: $s(x) \;\propto\; \log \frac{\hat{p}(x)}{p(x)}.$
However, if the score ordering coincides with the likelihood-ratio ordering,
i.e., $\arg\min_{j}\frac{\hat{p}_j}{p_j} = \arg\min_{j}\, s_j$,
then an adversary selecting tokens to minimize $\hat{p}_j/p_j$
simultaneously minimizes the detection score $s_j$, driving TPR to
baseline. Hence, maximizing statistical detectability necessarily
introduces the score-ratio alignment that adversaries can exploit.
\end{theorem}

\subsection{Implications for Watermark Design}
\label{subsec:design-implications}

Theorem~\ref{thm:tradeoff} establishes that the detectability-robustness tradeoff is not merely an engineering challenge that can be addressed,
but rather an {\bf information-theoretic impossibility} inherent to the current paradigm. This fundamental limitation has several practical implications.

\vspace{+0.05in}
\noindent\textbf{Robustness claims must specify the threat model.}
Prior works reporting robustness improvements (e.g., against paraphrasing
or bounded token editing~\cite{kuditipudi2024robust, Zhao2023ProvableRW}) do
so against threat models that do not include score-ratio-minimizing
adversaries. Our results show that schemes robust against post-hoc attacks
remain defeatable under generation-time evasion.
Any robustness claim must therefore explicitly state the adversary's access
model and attack phase.

\vspace{+0.05in}
\noindent\textbf{Improving detectability worsens robustness.}
Schemes that increase the logit boost $\delta$ or green-list fraction
$\gamma$ to improve detection power simultaneously sharpen the probability-ratio gap between green and red tokens, making the score-ratio ordering
more informative to the adversary. Adaptive and token-specific variants
that locally maximize $\delta_t$ or $\gamma_t$ are therefore more
exploitable at each step, not less.

\vspace{+0.05in}
\noindent\textbf{Breaking score-ratio alignment requires paradigm change.}
Theorem~\ref{thm:tradeoff} implies that robustness against our
attack
requires breaking the alignment between the detection score
ordering and the likelihood-ratio ordering. Two directions are promising:
(i) \emph{secret score randomization}, in which the score assigned to each
token is perturbed by a secret-key-dependent offset unknown to the
adversary, decoupling the observable probability ratio from the true
detection contribution; and (ii) \emph{multi-key ensemble watermarking},
in which detection aggregates scores across multiple independent watermark
keys, so that minimizing any single key's score ordering does not minimize
the ensemble statistic. We leave the formal analysis of these directions to
future work.

%% file: TheoreticalEffectiveness.tex
\section{Theoretical Effectiveness of the Attack}
\label{sec:attack-effectiveness}

We establish that score-ratio minimization reduces the
watermark statistic relative to the unwatermarked null
distribution. We first prove an expectation-level guarantee in  \S\ref{subsec:expected-score-reduction} that exact minimum-ratio selection reduces the expected
detector statistic. 
We then strengthen this conclusion to
first-order stochastic dominance, which yields a TPR bound at
every detection threshold in \S\ref{sec:tpr-dominance}.
We finally extend this exact result to
our practical finite-sample attack
(Algorithm \ref{alg:blackbox_watermark_attack})
in
\S\ref{subsec:practical-effectiveness}.
{\bf All proofs are deferred to Appendix \ref{app:proof-effectiveness}}.

\subsection{Expected-Score Reduction}
\label{subsec:expected-score-reduction}

At decoding step $t$, let $h$ denote the current prefix; and $h_{1:t-1}$ be the context until step $t-1$.
Define $\mathcal{S}_t(h_{1:t-1})
 :=
 \{i\in\mathcal{V}:p_t(i\mid h_{1:t-1})>0\}$
as the support of the base next-token distribution. Let
$s_t(i\mid h_{1:t-1},\xi_t)$ denote the oriented detection score of
token $i$ under key realization $\xi_t$.

We first make the following standard fresh-key assumption.

\begin{assumption}[Fresh-key null invariance]
For every admissible prefix $h$, let
$\xi_t$ be a fresh pseudorandom key  and let
$Y_t\sim p_t(\cdot\mid h)$ independently of $\xi_t$. The null
score $Z_t^0=\phi(\xi_t,h,Y_t)$
has a distribution that does not depend on $h$.
Further, the key values are independent across distinct
contexts.
\label{ass:nullinv}
\end{assumption}

\begin{theorem}[Effectiveness of minimum-ratio selection]
\label{thm:minimum-ratio-effectiveness}
If the score-ratio alignment holds pointwise at every
step
$t$, i.e.,
\begin{align}
 \arg\min_{i\in\mathcal{V}}
 \frac{\hat p_t(i\mid h_{1:t-1})}{p_t(i\mid h_{1:t-1})}
 \subseteq
 \arg\min_{i\in\mathcal{V}}
 s_t(i\mid h_{1:t-1},\xi_t).
 \label{eq:effectiveness-alignment}
\end{align}
Let the exact minimum-ratio adversary select
 $
 x_t^{\mathrm A}
 \in
 \arg\min_{i\in\mathcal{V}}
 \frac{\hat p_t(i\mid h_{1:t-1})}
      {p_t(i\mid h_{1:t-1})}.
$
Then, conditional on the attacked prefix and the current key,
\begin{equation}
 s_t(x_t^{\mathrm A})
 \leq
 \mathbb{E}
 \left[
   s_t(x_t^0)
   \mid h_{1:t-1},\xi_t
 \right],
 \label{eq:per-step-expectation-bound}
\end{equation}
where
$
 x_t^0
 \sim
 p_t(\cdot\mid h_{1:t-1})
$
is a counterfactual unwatermarked token sampled under the same
prefix.

Consequently, for any nonnegative weights
$w_1,\ldots,w_T$, the additive attacked score $ s_T^{\mathrm A}
 =
 \sum_{t=1}^{T}w_t s_t(x_t^{\mathrm A})$
satisfies
\begin{equation}
 \mathbb{E}[s_T^{\mathrm A}]
 \leq
 \mathbb{E}[s_T^0],
 \label{eq:cumulative-expectation-bound}
\end{equation}
under the fresh-key null-invariance assumption \ref{ass:nullinv}.

\end{theorem}

\begin{remark}
The inequality is strict whenever, with positive probability,
the base distribution assigns nonzero mass to a token whose
score is strictly larger than the minimum. Therefore, except
for the degenerate case in which the base model is supported
entirely on score-minimizing tokens, the attack strictly reduces
the expected watermark score.
\end{remark}

{\subsection{Stochastic Dominance and TPR Reduction}
\label{sec:tpr-dominance}

We now strengthen the expectation-level result of
Theorem~\ref{thm:minimum-ratio-effectiveness} to a distributional guarantee.
Throughout this section, token scores are oriented so that larger
values indicate stronger watermark evidence. Thus, for a scheme
whose original score has the opposite orientation, we negate the
score. For example, the oriented ITS score is
$Z_t=-s_t$.

Let $D_T=\Phi_T(Z_1,\ldots,Z_T)$
denote the detector statistic, where $\Phi_T$ is coordinatewise
nondecreasing. This includes sums, means, and positively
standardized sums.
For two real-valued random variables $U$ and $V$, we write
$U\preceq_{\mathrm{st}}V$ if $U$ is first-order stochastically
dominated by $V$, i.e., $\Pr(U\geq u)\leq\Pr(V\geq u), \forall u\in\mathbb{R}.$

At decoding step $t$, let $h\in\mathcal{V}^{t-1}$ be the current
prefix, let $p_t(\cdot\mid h)$ and $\hat p_t(\cdot\mid h)$ be the
base and watermarked next-token distributions, and define $R_t(i\mid h)
    :=
    \frac{\hat p_t(i\mid h)}{p_t(i\mid h)}$
for every $i$ in the support of $p_t(\cdot\mid h)$.

\begin{theorem}[Stochastic dominance under score-ratio minimization]
\label{thm:stochastic-dominance}
Suppose that score-ratio alignment holds pointwise: for every
step $t$, prefix $h$, and key realization $\xi_t$,
\begin{equation}
 \arg\min_{i:p_t(i\mid h)>0} R_t(i\mid h)
 \subseteq
 \arg\min_{i:p_t(i\mid h)>0}
       \phi(\xi_t,h,i).
 \label{eq:pointwise-alignment}
\end{equation}
Let the adversary select
$
x_t^{\mathrm A}
 \in
 \arg\min_{i:p_t(i\mid h_{1:t-1})>0}
 R_t(i\mid h_{1:t-1})
$
at every step $t$. Then $D_T^{\mathrm A}
    \preceq_{\mathrm{st}}
    D_T^0,$
where $D_T^0$ denotes the detector statistic under the
unwatermarked null distribution.

Further, if $\tau_\varepsilon$ is calibrated so that
$    \Pr_0(D_T^0\geq\tau_\varepsilon)\leq\varepsilon,
$
then the true positive rate on attacked watermarked text
satisfies
\begin{equation}
    \operatorname{TPR}_{\mathrm A}(\varepsilon)
    :=
    \Pr_{\mathrm A}
       (D_T^{\mathrm A}\geq\tau_\varepsilon)
    \leq\varepsilon.
 \label{eq:attacked-tpr-bound}
\end{equation}
\end{theorem}

\begin{remark}[Expectation minimization is insufficient]
The conclusion of Theorem~\ref{thm:stochastic-dominance}
cannot be obtained only from
$
    \mathbb{E}[D_T^{\mathrm A}]
    \leq
    \mathbb{E}[D_T^0].
$
Ordering the means generally does not  order the upper-tail
probabilities. The pointwise alignment condition is essential:
it produces an almost-sure coupling, which implies stochastic
dominance and therefore orders the TPR at every threshold.
\end{remark}}

\subsection{Practical Top-$K$ and Finite-Sample Attack}
\label{subsec:practical-effectiveness}

Algorithm~\ref{alg:blackbox_watermark_attack} operates on a top-$K$ candidate set and estimates
the watermarked probabilities using finitely many queries. Let
$C_t$ be the retained candidate set and define its excluded
base-model mass as
$
 \rho_t
 :=
 1-\sum_{i\in C_t}p_t(i\mid h_{1:t-1}).
$
Let $\zeta_t$ upper-bound the conditional probability that
finite-sample estimation, or filtering
causes the algorithm to select a token that is not a
score-minimizer within $C_t$.

\begin{theorem}[Approximate TPR guarantee]
\label{cor:practical-tpr}
If $\rho_t$ and $\zeta_t$ are valid conditional upper bounds
at every decoding step $t$, then
\begin{equation}
 \operatorname{TPR}_{\mathrm A}(\varepsilon)
 \leq
 \min\left\{
 1,\,
 \varepsilon+
 \sum_{t=1}^{T}(\rho_t+\zeta_t)
 \right\}.
 \label{eq:practical-tpr-bound}
\end{equation}
\end{theorem}

\begin{remark}
Theorem~\ref{cor:practical-tpr} explicitly isolates the two
sources of approximation in Algorithm~\ref{alg:blackbox_watermark_attack}: candidate truncation captured by $\rho_t$, and ratio-estimation or filtering error,
captured by $\zeta_t$. When both quantities are small, the
practical attack approaches the exact stochastic-dominance
guarantee.
\end{remark}

%% file: Experiments.tex
\section{Experimental Evaluation}
\label{sec:experiments}

Our evaluations mainly address the below
research questions:

\begin{itemize}[leftmargin=*,nosep]
    \item \textbf{RQ1.} Does the score-ratio minimization attack reduce
    watermark detection to baseline levels?

    \item \textbf{RQ2.} Does the attack preserve text quality relative to unwatermarked generation?

    \item \textbf{RQ3.} Are results consistent across model scales?

    \item \textbf{RQ4.} How does the hyperparameters affect our attack?
\end{itemize}

\subsection{Experimental Setup}
\label{subsec:setup}

\vspace{+0.05in}
\noindent\textbf{Watermarking schemes.}
We evaluate seven watermarking methods spanning both distortional and
distortion-free families shown in Table \ref{tab:scheme-summary}\footnote{SynthID requires intensive offline training on a private dataset to learn its Bayesian detector. Since this detector is unavailable in our threat model, we evaluate our attack against SynthID using its mean-score detector.}.

\noindent\textbf{Base and surrogate LLMs.}
We use OPT-1.3B~\cite{Zhang2022OPTOP} as the target base model on which the watermarking schemes are applied. As the surrogate model used in our attack (see Algorithm \ref{alg:blackbox_watermark_attack}), 
we three LLMs 
spanning a range of scales:
\begin{itemize}[leftmargin=*,nosep]
    \item \textbf{OPT-1.3B}~\cite{Zhang2022OPTOP}: a 1.3-billion-parameter
    model used to evaluate our attack to a small model and as as reference.

    \item \textbf{LLaMA-7B}~\cite{touvron2023llama}: a 7-billion-parameter
    decoder-only transformer used as the primary evaluation
    model.
    \item \textbf{LLaMA-13B}~\cite{touvron2023llama2openfoundation}: a 13-billion-parameter
    variant used to assess scale sensitivity.

\end{itemize}

\vspace{+0.05in}
\noindent\textbf{Datasets.}
Text generation is conditioned on prompts drawn from three corpora:
\begin{itemize}[leftmargin=*,nosep]
    \item \textbf{C4}~\cite{Raffel2019ExploringTL}: a large-scale general web text
    corpus used as the primary benchmark, consistent with prior
    watermarking evaluations~\cite{kirchenbauer2023watermark,
    huo2024tokenspecific}.
    \item \textbf{FakeNews}: a collection of news-style long-form
    narratives, used to evaluate performance on factual, structured prose.
    \item \textbf{Eli-5}~\cite{fan2019eli5}: a community  QA
    corpus, used to assess performance on open-ended,
    explanatory generation.
\end{itemize}

\vspace{+0.05in}
\noindent\textbf{Attack configuration.}
We evaluate our attack against 100 watermarked texts under the considered threat model in \S\ref{sec:attack}. 
By default, we use $N=256$ queries per generation step, {$K=40$ in top-K candidate set, and generation length $T=256$. We also study the impact of these hyperparameters.}

\vspace{+0.05in}
\noindent\textbf{Evaluation metrics.}
We report the following metrics:

\begin{itemize}[leftmargin=*,nosep]
    \item \textbf{TPR @ FPR=1\%}: the fraction of
    adversarially generated texts correctly classified as watermarked by
    the detector at a fixed false positive rate of $1\%$. {\bf Lower TPR
    indicates more effective watermark evasion.}

    \vspace{+0.05in}
    \item \textbf{$p$-value}: the detector's hypothesis
    test applied to adversarially generated text.
    {\bf Values above $0.05$
    indicate failure to reject the null hypothesis (non-watermarked),
    confirming evasion.}

    \item \textbf{Z-score}: the normalized detection statistic. {\bf Negative
    or near-zero Z-scores indicate that the watermark signal has been
    suppressed to below the detection threshold.}

    \item \textbf{Perplexity (PPL)}: text fluency measured by GPT-2~\cite{Radford2019LanguageMA}. We report PPL to verify that evasion
    does not degrade output quality.
    {\bf Lower perplexity indicates more
    fluent and natural text.}
\end{itemize}

\begin{table*}[t]
  \centering

  \caption{
  Attack results on three datasets using OPT-1.3B (base model) as the surrogate model. TPR is reported at a 1\% false positive rate. Larger $p$-values and lower TPRs indicate stronger evasion. A $p$-value above $0.05$ indicates unsuccessful watermark detection. PPL denotes the perplexity of adversarially generated text; lower values correspond to more fluent outputs.}
 \label{tab:opt1.3b-all}

  \renewcommand{\arraystretch}{0.99}
  \setlength{\tabcolsep}{4pt}

  \resizebox{\textwidth}{!}{
  \begin{tabular}{lcccc|cccc|cccc}
    \toprule
    & \multicolumn{4}{c|}{\textbf{C4}}
    & \multicolumn{4}{c|}{\textbf{FakeNews}}
    & \multicolumn{4}{c}{\textbf{Eli-5}} \\
    \cmidrule(lr){2-5}
    \cmidrule(lr){6-9}
    \cmidrule(lr){10-13}

    \textbf{Scheme}
    & \textbf{TPR}$\downarrow$
    & \textbf{$p$-value}$\uparrow$
    & \textbf{$Z$-score}$\downarrow$
    & \textbf{PPL}$\downarrow$
    & \textbf{TPR}$\downarrow$
    & \textbf{$p$-value}$\uparrow$
    & \textbf{$Z$-score}$\downarrow$
    & \textbf{PPL}$\downarrow$
    & \textbf{TPR}$\downarrow$
    & \textbf{$p$-value}$\uparrow$
    & \textbf{$Z$-score}$\downarrow$
    & \textbf{PPL}$\downarrow$ \\
    \midrule

    KGW
    & 0.01 & 0.98 & -4.61 & 14.72
    & 0.01 & 0.98 & -4.23 & 15.71
    & 0.02 & 0.98 & -4.39 & 15.16 \\

    Token-Specific KGW
    & 0.01 & 0.97 & -4.18 & 5.91
    & 0.01 & 0.97 & -3.47 & 6.35
    & 0.01 & 0.96 & -2.13 & 6.48 \\

    Adaptive KGW
    & 0.03 & 0.24 & 0.73 & 3.86
    & 0.02 & 0.34 & 0.29 & 4.54
    & 0.03 & 0.31 & 0.46 & 4.13 \\

    GumbelSoft
    & 0.02 & 0.49 & -0.06 & 4.51
    & 0.01 & 0.63 & -0.36 & 4.92
    & 0.01 & 0.61 & -0.38 & 4.94 \\

    DipMark
    & 0.01 & 0.74 & -0.63 & 5.04
    & 0.02 & 0.66 & -0.41 & 5.27
    & 0.01 & 0.64 & -0.24 & 5.61 \\

    ITS
    & 0.02 & 0.27 & 0.69 & 5.21
    & 0.01 & 0.21 & 0.81 & 4.96
    & 0.03 & 0.43 & 0.12 & 3.91 \\

    SynthID
    & 0.01 & 0.93 & -2.41 & 5.78
    & 0.01 & 0.98 & -3.61 & 4.24
    & 0.01 & 0.98 & -2.87 & 5.19 \\

    \bottomrule
  \end{tabular}
  }
\end{table*}

\begin{table*}[!t]
\centering

\caption{Attack results on three datasets using LLaMA-7B as the surrogate model.
}
\label{tab:llama7b-all}
\renewcommand{\arraystretch}{0.99}
\setlength{\tabcolsep}{4pt}

\resizebox{\textwidth}{!}{
\begin{tabular}{lcccc|cccc|cccc}
\toprule

& \multicolumn{4}{c|}{\textbf{C4}}
& \multicolumn{4}{c|}{\textbf{FakeNews}}
& \multicolumn{4}{c}{\textbf{Eli-5}} \\

\cmidrule(lr){2-5}
\cmidrule(lr){6-9}
\cmidrule(lr){10-13}

\textbf{Scheme}
& \textbf{TPR}$\downarrow$
& \textbf{$p$-value}$\uparrow$
& \textbf{Z-score}$\downarrow$
& \textbf{PPL}$\downarrow$

& \textbf{TPR}$\downarrow$
& \textbf{$p$-value}$\uparrow$
& \textbf{Z-score}$\downarrow$
& \textbf{PPL}$\downarrow$

& \textbf{TPR}$\downarrow$
& \textbf{$p$-value}$\uparrow$
& \textbf{Z-score}$\downarrow$
& \textbf{PPL}$\downarrow$
\\

\midrule

KGW
& 0.00 & 1.00 & -4.94 & 13.18
& 0.00 & 0.99 & -4.67 & 13.99
& 0.03 & 0.99& -4.67 & 13.86
\\

Token-Specific KGW
& 0.00 & 1.00 & -5.37 & 6.23
& 0.00 & 0.99 & -3.96 & 6.23
& 0.00 & 0.98 & -2.27 & 6.56
\\

Adaptive KGW
& 0.01 & 0.07 & 1.47 & 3.19
& 0.00 & 0.058 & 1.57 & 4.42
& 0.01 & 0.28 & 0.58 & 4.06
\\

GumbelSoft
& 0.00 & 0.46 & 0.10 & 4.39
& 0.00 & 0.65 & -0.40 & 4.98
& 0.00 & 0.66 & -0.42 & 4.89
\\

DipMark
& 0.00 & 0.68 & -0.48 & 4.95
& 0.00 & 0.59 & -0.25 & 5.15
& 0.00 & 0.61 & -0.29 & 5.27
\\

ITS
& 0.00 & 0.32 & 0.48 & 5.66
& 0.00 & 0.16 & 0.99 & 4.81
& 0.02 & 0.25 & 0.65 & 3.28
\\

SynthID
& 0.00 & 0.99 & -3.56 & 5.49
& 0.00 & 0.99 & -3.38 & 4.16
& 0.00 & 0.97 & -1.94 & 5.12
\\

\bottomrule
\end{tabular}
}
\end{table*}

\begin{table*}[!t]
\centering

\caption{Attack results on three datasets using LLaMA-13B as the surrogate model.}
\label{tab:llama13b-all}
\renewcommand{\arraystretch}{0.99}
\setlength{\tabcolsep}{4pt}

\resizebox{\textwidth}{!}{
\begin{tabular}{lcccc|cccc|cccc}
\toprule

& \multicolumn{4}{c|}{\textbf{C4}}
& \multicolumn{4}{c|}{\textbf{FakeNews}}
& \multicolumn{4}{c}{\textbf{Eli-5}} \\

\cmidrule(lr){2-5}
\cmidrule(lr){6-9}
\cmidrule(lr){10-13}

\textbf{Scheme}
& \textbf{TPR}$\downarrow$
& \textbf{$p$-value}$\uparrow$
& \textbf{Z-score}$\downarrow$
& \textbf{PPL}$\downarrow$

& \textbf{TPR}$\downarrow$
& \textbf{$p$-value}$\uparrow$
& \textbf{Z-score}$\downarrow$
& \textbf{PPL}$\downarrow$

& \textbf{TPR}$\downarrow$
& \textbf{$p$-value}$\uparrow$
& \textbf{Z-score}$\downarrow$
& \textbf{PPL}$\downarrow$
\\

\midrule

KGW
& 0.00 & 0.99 & -4.34 & 15.94
& 0.00 & 0.99 & -3.98 & 16.43
& 0.00 & 0.99& -4.12 & 16.37
\\

Token-Specific KGW
& 0.00 & 0.99 & -3.89 & 5.63
& 0.00 & 0.99 & -3.14 & 6.47
& 0.00 & 0.97 & -1.98 & 6.56
\\

Adaptive KGW
& 0.02 & 0.28 & 0.57 & 4.23
& 0.00 & 0.56 & -0.17 & 4.65
& 0.02 & 0.28 & 0.58 & 4.06
\\

GumbelSoft
& 0.01 & 0.43 & 0.17 & 4.62
& 0.00 & 0.69 & -0.50 & 4.86
& 0.00 & 0.66 & -0.42 & 4.89
\\

DipMark
& 0.00 & 0.80 & -0.86 & 5.12
& 0.00 & 0.62 & -0.31 & 5.36
& 0.00 & 0.57 & -0.18 & 5.93
\\

ITS
& 0.00 & 0.19 & 0.86 & 4.83
& 0.00 & 0.16 & 0.98 & 5.11
& 0.01 & 0.77 & -0.76 & 4.13
\\

SynthID
& 0.00 & 0.87 & -1.17 & 6.12
& 0.00 & 0.99 & -3.84 & 4.32
& 0.00 & 0.99 & -3.94 & 5.12
\\

\bottomrule
\end{tabular}
}
\end{table*}

\subsection{Results}
\label{subsec:results}

Table~\ref{tab:opt1.3b-all}--Table~\ref{tab:llama13b-all}
report TPR@FPR=1\%, $p$-value, Z-score, and PPL.
We summarize the following key findings.

\vspace{+0.05in}
\noindent\textbf{Finding 1: Near-universal watermark evasion.}
Across all seven watermarking schemes and all model--dataset combinations,
the score-ratio minimization attack achieves a TPR of $0.00$--$0.05$,
reducing detection to or below the false positive rate baseline. This
confirms RQ1: the attack generalizes across diverse scheme designs,
including both distortional
(KGW, Token-Specific, Adaptive) and
distortion-free (GumbelSoft, DipMark, ITS, SynthID) families. 
The result is consistent with the score-ratio alignment
theorems of \S\ref{sec:optimality}: since alignment holds for every scheme,
the attack is provably optimal in each case, and the empirical TPR results
follow from the theory.

\vspace{+0.05in}
\noindent\textbf{Finding 2: Text quality is preserved.}
Perplexity remains low and competitive with unwatermarked generation across
all conditions (e.g., PPL $<7$ in all cases except on KGW). This confirms RQ2: the attack is genuinely
non-distortionary in practice, not only in theory. The preservation of
fluency validates that the score-ratio minimization strategy selects tokens
that are both low-scoring and semantically coherent, consistent with the
distributional proximity argument of \S\ref{sec:attack}.

\vspace{+0.05in}
\noindent\textbf{Finding 3: Results are consistent across model scales.}
Comparing results on LLaMA-7B and LLaMA-13B, 
we observe no
meaningful difference in attack effectiveness since TPR is very small (most are 0.00)
for both model scales across all schemes. This confirms RQ3 and suggests
that vulnerability is not a function of model capacity but of the
structural alignment property.

\vspace{+0.05in}
\noindent\textbf{Finding 4: Statistically significant evasion.}
The $p$-values for all
schemes substantially exceed the standard
$0.05$ significance threshold, confirming that the detector cannot
reject the null hypothesis (non-watermarked) for adversarially generated
text. Z-scores are negative or near-zero in almost all cases, indicating that the
cumulative detection statistic has been driven below the detection
threshold. These results are consistent with Theorem~\ref{thm:defeatable}:
minimizing $S = \sum_t s_{i_t^*}$ at each step drives $\mathbb{E}[S]$
below the threshold $\tau$, causing the TPR to collapse.

\begin{figure*}[!t]
    \centering
    \begin{subfigure}[b]{0.32\linewidth}
        \centering
        \includegraphics[width=\linewidth,height=3.5cm]{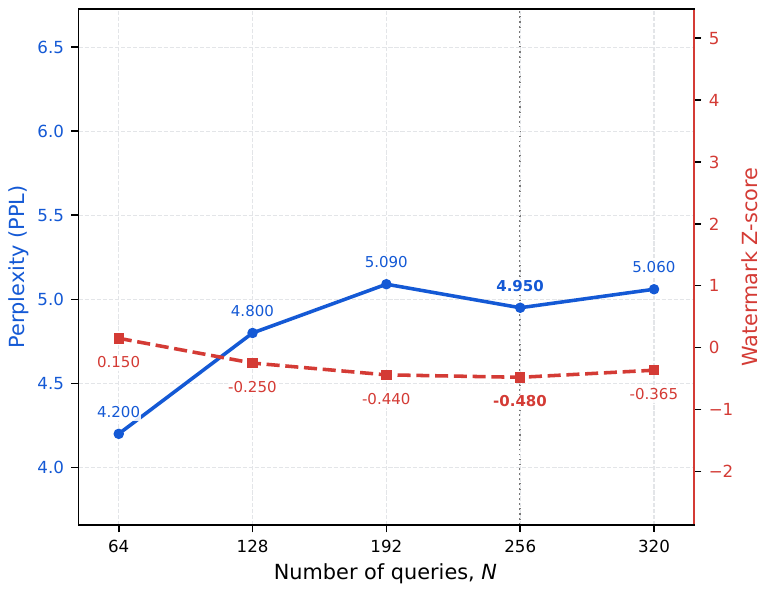}
        \caption{C4}
        \label{fig:ppl-zscore-C}
         \end{subfigure}
    \hfill
    \begin{subfigure}[b]{0.32\linewidth}
        \centering
        \includegraphics[width=\linewidth,height=3.5cm]{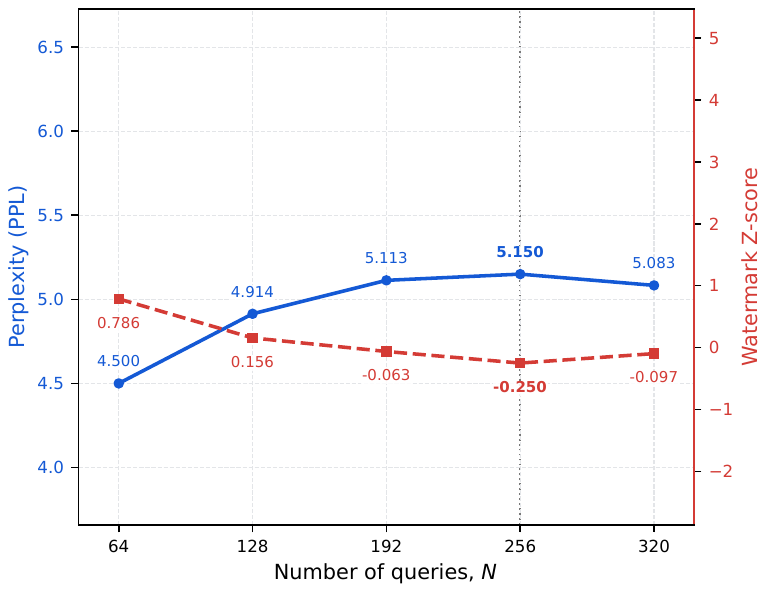}
        \caption{Fake-News}
        \label{fig:ppl-zscore-B}
    \end{subfigure}
    \hfill
    \begin{subfigure}[b]{0.32\linewidth}
        \centering
        \includegraphics[width=\linewidth,height=3.5cm]{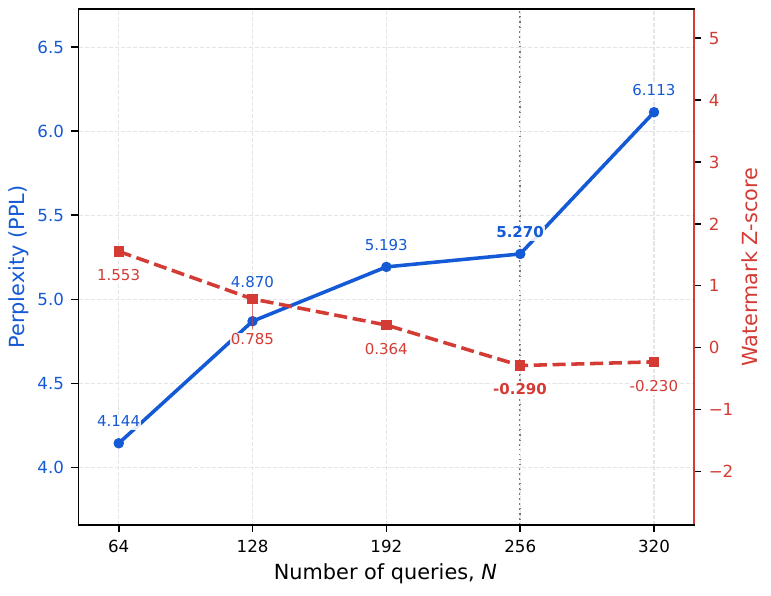}
        \caption{Eli-5}
        \label{fig:ppl-zscore-A}
    \end{subfigure}
    \caption{Perplexity (left axis) and watermark Z-score (right axis) as a function of \#queries $N$.
    }
    \label{fig:ppl-zscore-multi}

\end{figure*}

\begin{figure*}[!t]
    \centering
        \begin{subfigure}[b]{0.32\linewidth}
        \centering
        \includegraphics[width=\linewidth]{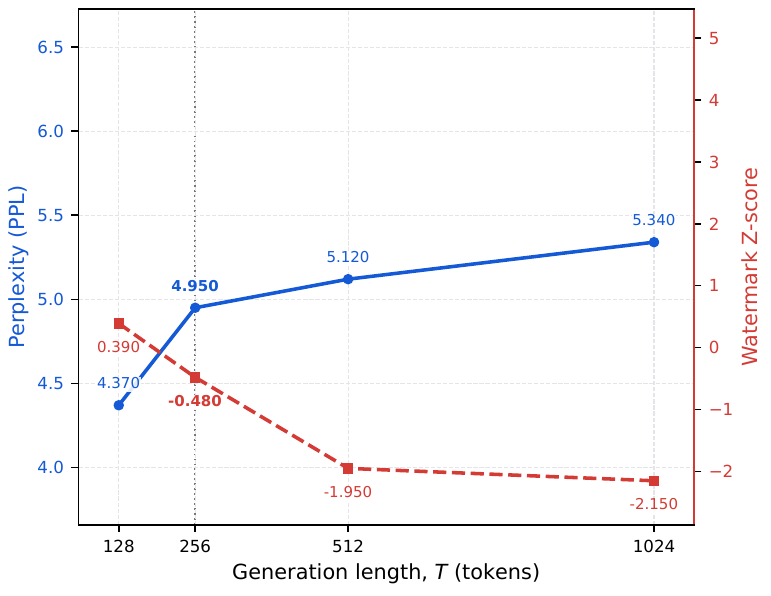}
        \caption{C4}
        \label{fig:ppl-zscore-C}
    \end{subfigure}
    \hfill
    \begin{subfigure}[b]{0.32\linewidth}
        \centering
        \includegraphics[width=\linewidth]{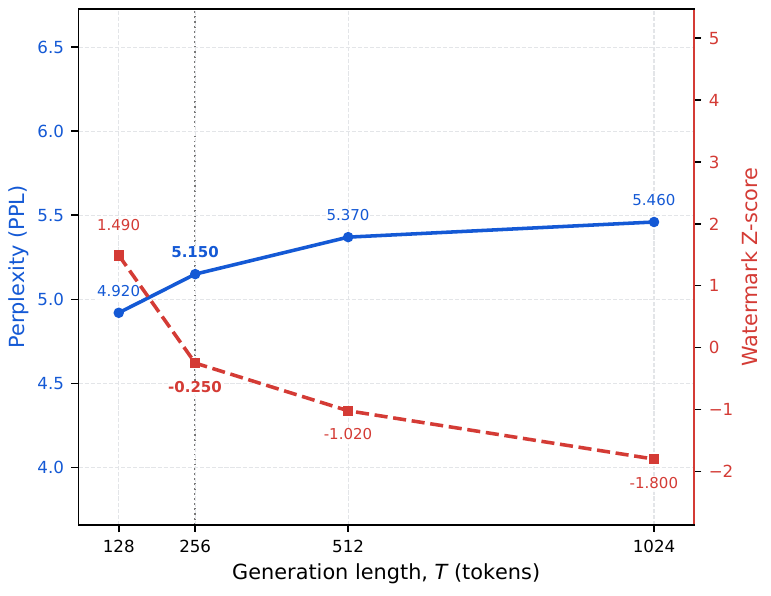}
        \caption{Fake-News}
        \label{fig:ppl-zscore-B}
    \end{subfigure}
    \hfill
    \begin{subfigure}[b]{0.32\linewidth}
        \centering
        \includegraphics[width=\linewidth]{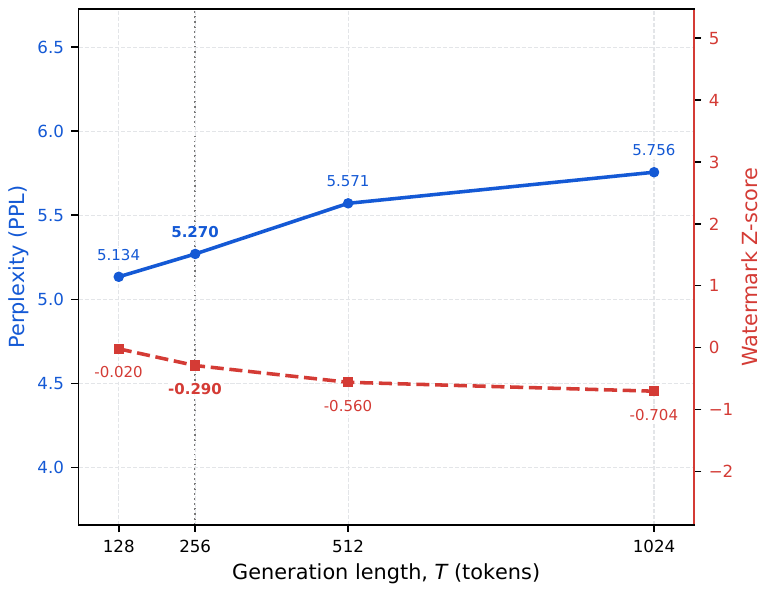}
        \caption{Eli-5}
        \label{fig:ppl-zscore-A}
    \end{subfigure}
    \caption{Effect of generated text length $T$ (tokens) on attack
effectiveness and text quality.
}
    \label{fig:ppl-zscore-queries}

\end{figure*}

\begin{figure*}[!t]
    \centering
        \begin{subfigure}[b]{0.32\linewidth}
        \centering
        \includegraphics[width=\linewidth]{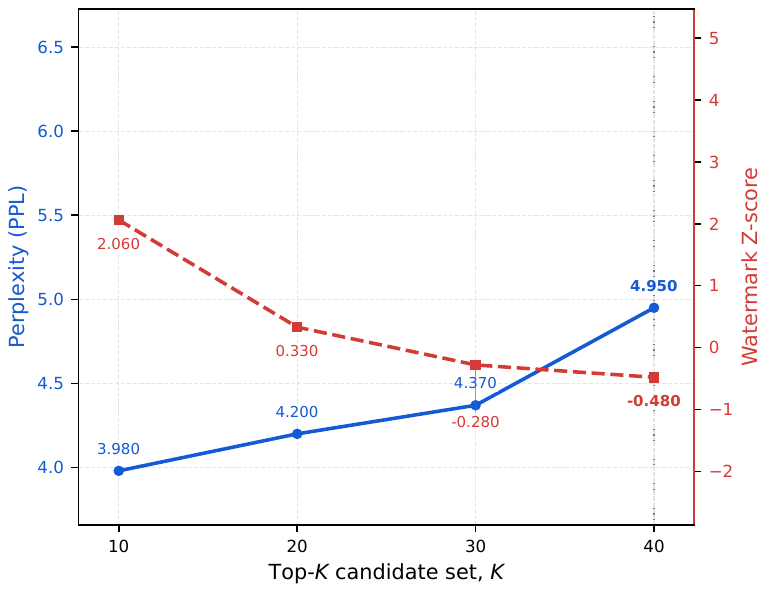}
        \caption{C4}
        \label{fig:ppl-zscore-C}
    \end{subfigure}
    \hfill
    \begin{subfigure}[b]{0.32\linewidth}
        \centering
        \includegraphics[width=\linewidth]{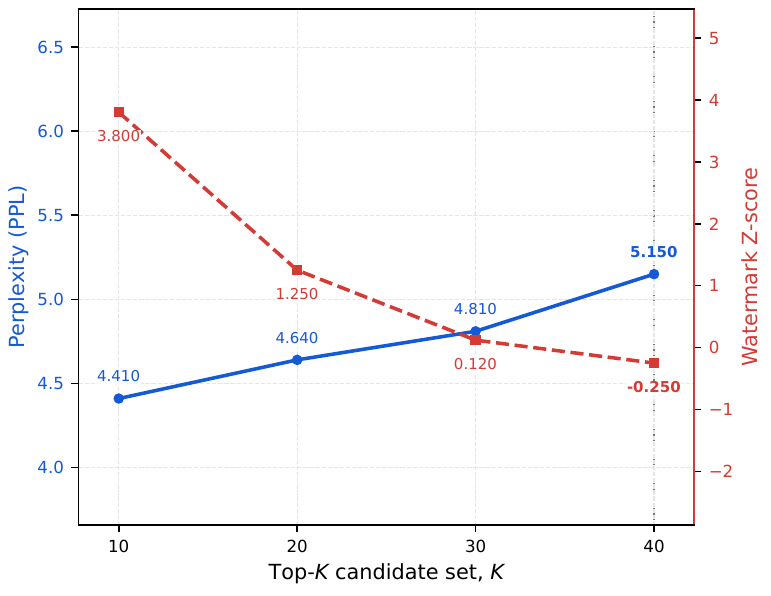}
        \caption{Fake-News}
        \label{fig:ppl-zscore-B}
    \end{subfigure}
    \hfill
    \begin{subfigure}[b]{0.32\linewidth}
        \centering
        \includegraphics[width=\linewidth]{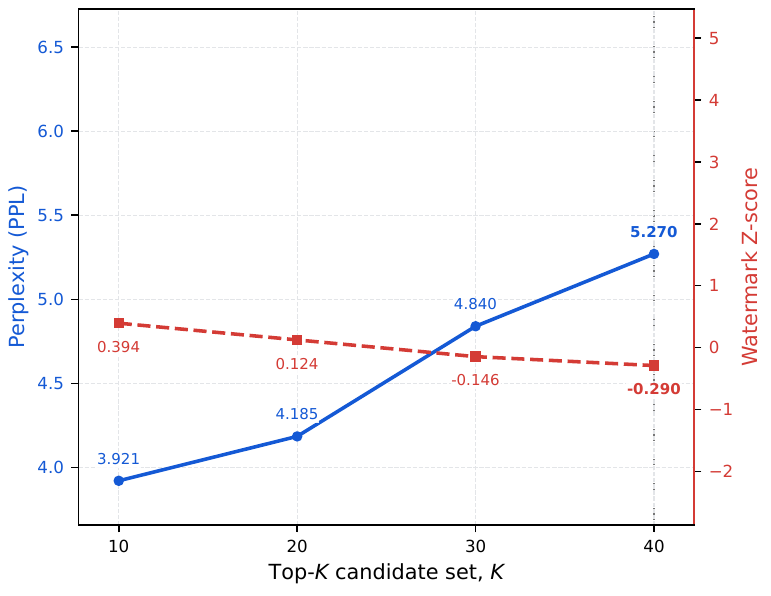}
        \caption{Eli-5}
        \label{fig:ppl-zscore-A}
    \end{subfigure}
    \caption{Effect of $K$ in Top-K candidate set on attack
effectiveness and text quality.}
    \label{fig:impact-top-k}

\end{figure*}

\subsection{Ablation Study}
\label{subsec:ablation}

In this section, we investigate the impact of key parameters on attack performance (RQ4): query budget $N$, generated text length $T$, and $K$ in Top-K candidate tokens. For brevity, we report results using the medium LLaMA-7B model on DipMark.

\vspace{+0.05in}
\noindent\textbf{Impact of $N$.}
Figure~\ref{fig:ppl-zscore-multi} studies the effect of the query budget $N$ on attack effectiveness and text quality. Across all datasets, increasing $N$ consistently decreases the watermark Z-score, indicating stronger watermark evasion. For example, on C4, the Z-score decreases from $0.15$ to $-0.48$ as $N$ increases, while the perplexity increases only slightly, from $4.20$ to $5.09$, suggesting that the attack preserves text quality well even with a larger query budget.

\vspace{+0.05in}
\noindent\textbf{Impact of $T$.}
Figure~\ref{fig:ppl-zscore-queries} shows the impact of the generated text length $T$. As $T$ increases, the watermark Z-score becomes increasingly negative, indicating stronger watermark evasion. This trend is expected, as longer generations provide more evidence for the attack to identify and suppress the watermark signal. Meanwhile, perplexity increases marginally with $T$, suggesting a modest trade-off in text quality as the attack operates over longer generations.

\vspace{+0.05in}
\noindent \textbf{Impact of $K$.}
Figure~\ref{fig:impact-top-k} shows the results with varying $K$ candidate tokens in Algorithm~\ref{alg:blackbox_watermark_attack}. We observe that increasing $K$ decreases the watermark Z-score, indicating stronger attack performance; and the perplexity increases since there is a trade-off between detectability and quality.

%% file: RelatedWork.tex
\section{Related Work}
\label{sec:relatedwork}

\subsection{AI-Generated Text Detection}

\subsubsection{Non-watermark based Detection}
A large body of work develops detectors that do not require a watermark signal to be embedded at generation time.

\vspace{+0.05in} 
\noindent {\bf Statistical zero-shot detectors.}
Mitchell et al.~\cite{mitchell2023detectgpt} propose
\emph{DetectGPT}, which exploits the observation that
AI-generated text tends to lie in regions of the
probability surface with negative curvature; detection is
performed by checking whether small random perturbations of
a candidate text consistently reduces its log-probability under
the model.
Su et al.~\cite{su2023detectllm} extend the work 
by
using the \emph{log rank} of sampled tokens—rather than
log-probability curvature—as a detection signal, improving
efficiency and accuracy.
Bao et al.~\cite{bao2024fastdetectgpt} introduce
\emph{Fast-DetectGPT}, which replaces  perturbation-based
curvature estimate with a conditional probability curvature
measure, achieving  accuracy comparable to DetectGPT but with less 
costs.
Hans et al.~\cite{hans2024binoculars} propose
\emph{Binoculars}, a zero-shot detector that compares the
perplexity of a text under two closely related models, 
without requiring
access to the original generator.

\vspace{+0.05in} \noindent {\bf Learning-based detectors.}
Guo et al.~\cite{guo2023closechatgpt} train 
supervised classifiers to distinguish GPT-generated text from
human text, finding that finetuned detectors can achieve high
accuracy on in-distribution data but degrade under distribution
shift.
He et al.~\cite{he2024mgtbench} introduce MGTBench, a 
benchmark for evaluating machine-generated text
detectors across diverse generation models, domains, and
adversarial conditions.
Ghosal et al.~\cite{ghosal2023towards} provide a survey of
possibilities and impossibilities of AI-generated text
detection, including both theoretical limits and practical
evaluations.

\subsubsection{Watermark based Detection}

\noindent {\bf Distortional watermarking.}
Distortional schemes \cite{kirchenbauer2023watermark,huo2024tokenspecific,liu2024adaptive,fernandez2023three,hu2024unbiased} explicitly modify the token distribution,
accepting a measurable divergence from the base model in exchange
for a stronger detection signal.
The foundational work~\cite{kirchenbauer2023watermark}
introduces the \emph{KGW green/red list} watermark, which
pseudorandomly partitions the vocabulary into green- and
red-listed tokens at each decoding step and applies a fixed logit
boost to all green tokens; detection is performed via a
$Z$-score on the fraction of green tokens observed in a candidate
text.
Huo et al.~\cite{huo2024tokenspecific} introduce
\emph{token-specific} watermarking, which assigns
context-dependent boost magnitudes $\delta_i$ to individual
tokens, improving fluency while preserving detectability.
Liu and Bu~\cite{liu2024adaptive} propose an \emph{adaptive}
watermark that dynamically adjusts both the green-list fraction
and the boost magnitude based on
real-time semantic coherence signals, reducing perceptual
artifacts in sensitive decoding contexts.
While these refinements improve the detectability-quality
tradeoff, we show in \S\ref{sec:optimality} that all KGW-style
schemes share a structural vulnerability: their detection score
functions are aligned with the token probability-ratio ordering,
making them susceptible to our unified attack framework.

\vspace{+0.05in} \noindent {\bf Distortion-free watermarking.}
Distortion-free schemes \cite{christ2024undetectable,fu2024gumbelsoft,aaronson2022watermarking,kuditipudi2024robust,wu2024dipmark,liu2024semantic,Zhao2023ProvableRW,takezawa2023necessary} aim to preserve the original output
distribution in expectation for all tokens, 
while still
embedding a detectable signal~\cite{christ2024undetectable}.
GumbelSoft~\cite{fu2024gumbelsoft} achieves distortion-free
watermarking by injecting Gumbel-distributed pseudorandom noise
into logits before softmax sampling; the detector recovers the
Gumbel noise component at the sampled token and aggregates these
values into a $Z$-score.
Kuditipudi et al.~\cite{kuditipudi2024robust} propose a
distortion-free scheme ITS with formal robustness guarantees against
bounded text editing, though their analysis does not address
adaptive adversaries that exploit score structure.
Wu et al.~\cite{wu2024dipmark} introduce {DiPMark}, which
reweights the token distribution via a random permutation, assigning binary colour scores to tokens according to their permuted rank.
SynthID~\cite{dathathri2024scalable} is the first production-ready generative LLM watermarking system, designed to be efficient, non-invasive, and detectable at scale. Its core contribution is a novel \emph{Tournament Sampling} mechanism for next-token generation, which achieves state-of-the-art watermark detectability.

\subsection{Attacks on AI-Generated Text Detectors}

\noindent {\bf Post-hoc attacks.}
Most existing 
attacks 
operate after text generation by modifying the generated content to disrupt the detector. 
The most widely studied approach is paraphrasing, where an auxiliary  model rewrites the  
generated text while attempting to preserve its semantic meaning~\cite{krishna2023paraphrasing}. Other post-hoc attacks include token-level editing~\cite{Piet2023MARKMyWORDSAA}, detector red-teaming using adversarially finetuned language models~\cite{shi2024red}, watermark-evasion generation strategies~\cite{jiang2023evading,meng2025gradescape}, and universal adversarial triggers that manipulate model outputs through specially crafted prompts~\cite{wallace2019universal}. Beyond empirical attacks, Sadasivan et al.~\cite{Sadasivan2023CanAT} establish information-theoretic limitations of watermark detection under sufficiently powerful paraphrasing adversaries. Although these approaches can substantially reduce detection rates, they typically exhibit an inherent quality--evasion tradeoff: weakening the watermark often requires modifying the generated text, which may introduce semantic drift, stylistic changes, or additional computational overhead.

\vspace{+0.05in} \noindent {\bf Generation-time evasion.}
These works consider attacks that 
intervene at generation time
rather than after the fact.
Zhao et al.~\cite{Zhao2023ProvableRW} study an adversary that
selects among multiple watermarked continuations, but assume
access to the detector score 
during generation.
Zhang et al.~\cite{zhang2024watermarks} 
show that no scheme
can simultaneously guarantee detectability and security against
a computationally bounded forger with access to the generator.
Abdelnabi and Fritz~\cite{abdelnabi2021adversarialwatermarking}
study adversarial attacks on \emph{text} watermarks in an
earlier generation of transformer-based systems, providing
an empirical baseline for generation-time adversaries.
Jovanovic et al.~\cite{jovanovic2024watermark} propose a watermark stealing attack against KGW~\cite{kirchenbauer2023watermark}, where an attacker queries the watermarked model to approximately reverse-engineer its secret watermarking rules and then exploits the learned watermark to facilitate watermark removal. Omidi et al.~\cite{omidi2026google,omidi2026google2} propose the first layer inflation attack against SynthID, which exploits its tournament-based sampling mechanism by externally adding additional tournament layers through black-box queries.
Our attack is the first to formulate generation-time evasion as a principled optimization problem, establish its optimality across multiple watermark families, and empirically achieve near-zero true positive rates without sacrificing text quality. 
Importantly, it requires only black-box sampling access, making the attack practical for API-level adversaries.

\subsection{Detectability–Robustness Tradeoffs}

The tension between watermark detectability and adversarial
robustness has been recognized informally in prior work but
not formally characterized for the LLM setting. 
Christ et al.~\cite{christ2024undetectable} show that
distortion-free watermarks are necessary for computational
undetectability, implying a tradeoff between quality
preservation and statistical detectability.
Zhang et al.~\cite{zhang2024watermarks} prove that no
watermark can be simultaneously unremovable and undetectable
under a computationally bounded model, establishing a
cryptographic impossibility analogous to our
information-theoretic result.
Takezawa et al.~\cite{takezawa2023necessary} characterize
necessary and sufficient conditions on score functions for a
watermark to achieve 
detectability and a form of
robustness.

Our work proves that maximizing detection power under
the Neyman–Pearson lemma~\cite{Neyman1933OnTP}
necessarily aligns the score function with the likelihood-ratio
ordering—the exact structure our adversary exploits.  This
result is stronger than prior informal observations: it
establishes that the tradeoff is not merely empirical but
information-theoretically fundamental within the current
generator–detector paradigm, and that no in-paradigm fix can
simultaneously achieve optimal detectability and resistance to
score-ratio-minimizing adversaries.

%% file: Conclusion.tex
\section{Conclusion}
\label{sec:conclusion}

 We develop a unified attack framework on watermarked LLMs based on \emph{score-ratio alignment}, a structural property shared by most generative watermarks. By selecting score-minimizing tokens at each decoding step, our attack achieves a near-zero true positive rate with minimal degradation in text quality. 
 We show that statistically optimal watermark detectors inherently align their detection scores with score-ratio ordering, making such watermarks fundamentally vulnerable to our attack. Our results reveal an inherent limitation of statistical generative watermarks  and establish score-ratio minimization 
as an important benchmark for evaluating watermark robustness.

%% file: Appendix.tex
\section{Proof of Optimality}
\label{app:proof-opt}

\subsection{Proof of Theorem \ref{thm:logit-addition}}

    Compute the probability ratio directly:
\begin{align}
    \frac{\hat{p}_j}{p_j}
    &= \frac{e^{\hat{l}_j} / \sum_k e^{\hat{l}_k}}
            {e^{l_j}      / \sum_k e^{l_k}}
     = \frac{e^{l_j + \delta\xi_j}}{e^{l_j}}
       \cdot \frac{\sum_k e^{l_k}}{\sum_k e^{l_k + \delta\xi_k}}
     = e^{\delta\xi_j} \cdot C,
     \label{eq:ratio-logit}
\end{align}
where $C := \frac{\sum_k e^{l_k}}{\sum_k e^{l_k + \delta\xi_k}}$

is a normalizing constant that does not depend on $j$. Hence,
$ \arg\min_{j \in \mathcal{V}} \frac{\hat{p}_j}{p_j}
    = \arg\min_{j \in \mathcal{V}} e^{\delta\xi_j}.$

Since $\delta > 0$ and $x \mapsto e^{\delta x}$ is strictly increasing,
\[
    \arg\min_{j \in \mathcal{V}} e^{\delta\xi_j}
    = \arg\min_{j \in \mathcal{V}} \xi_j
    = \arg\min_{j \in \mathcal{V}} \mathbb{E}[s_j],
\]
where the final equality holds because $s_j = \xi_j$ and the expectation
is taken over the PRF randomness conditioned on the context. This
establishes the stated inclusion.

\subsection{Proof of Theorem \ref{thm:kgw-variant}}

    We split into two cases based on whether any red token exists.

{First of all, we know that based on the design of KGW variant watermarks, we have $\delta_j>0$ .}

\noindent\textit{Case 1: At least one red token exists.}
Suppose $\xi_j = 0$ for some $j \in \mathcal{V}$. For any red token,
$e^{\delta_j \cdot 0} = 1$; for any green token, $e^{\delta_j \cdot 1}
= e^{\delta_j} > 1$ since $\delta_j > 0$. Therefore, the probability ratio
is minimized exactly by the red tokens:
\[
    \arg\min_{j \in \mathcal{V}} \frac{\hat{p}_j}{p_j}
    = \{j:  \xi_j = 0\}
    = \arg\min_{j \in \mathcal{V}} s_j.
\]

\medskip
\noindent\textit{Case 2: All tokens are green.}
If $\xi_j = 1$ for all $j \in \mathcal{V}$, then $s_j = 1$ for all $i$,
so $\arg\min_j s_j = \mathcal{V}$. Moreover, every token satisfies
$e^{\delta_j\xi_j} = e^{\delta_j}$, giving $ \frac{\hat{p}_j}{p_j} = e^{\delta_j} \cdot C,$

which is not necessarily equal across tokens (since $\delta_j$ may vary).
Nevertheless, the minimizer of $\hat{p}_j / p_j$ belongs to $\mathcal{V}
= \arg\min_j s_j$, so the inclusion holds.

\subsection{Proof of Theorem \ref{thm:kgw-feasibility}}

    If at least one red token exists, then for every red token $i$, $s_j =
\xi_j = 0$ is the minimum achievable score, and $e^{\delta_j \cdot 0} = 1
< e^{\delta_j} = e^{\delta_j \cdot 1}$ for any $\delta_j > 0$, confirming
that red tokens attain the smallest probability ratio. If no red token
exists, all tokens have identical score $s_j = 1$ and the minimizer
condition holds vacuously over $\mathcal{V}$.

\subsection{Proof of Theorem \ref{thm:dipmark}}

    We consider two exhaustive cases.

\medskip
\noindent\textit{Case 1: All tokens are green.}
If $s_{\iota[k]} = 1$ for all $k$, then $\mathbb{E}[s_j] = 1$ for all
$j \in \mathcal{V}$, so
$\arg\min_{j} \mathbb{E}[s_j] = \mathcal{V}$.
Any minimizer of $\hat{p}_j / p_j$ therefore belongs to $\mathcal{V}$,
and the inclusion holds trivially.

\medskip
\noindent\textit{Case 2: At least one token is red.}
Since the red prefix is non-empty, $\min_j s_j = 0$ and
$\min_j \mathbb{E}[s_j] = 0$. By the DipMark construction, the token
$\iota[0]$---the first in the shuffled order---has the highest probability
(over the random permutation) of landing in the red prefix, and therefore
attains the smallest expected score:

\[
    \iota[0] \in \arg\min_{j \in \mathcal{V}} \mathbb{E}[s_j].
\]

In DipMark's reweighting procedure, probability mass is redistributed from
tokens at the start of the permuted list toward those at the end: the
earliest tokens are the first to have their probabilities reduced (and
potentially
zero), while later tokens remain invariant or receive
increased mass. Consequently, the maximum degradation $\hat{p}_j / p_j$ is
achieved by the token appearing earliest in the permuted order:
\[
     \iota[0] {\in} \arg\min_{j \in \mathcal{V}} \frac{\hat{p}_j}{p_j}.
\]
Combining both inclusions:
\[
    \iota[0] {\in} \arg\min_{j \in \mathcal{V}} \frac{\hat{p}_j}{p_j}
    \in \arg\min_{j \in \mathcal{V}} \mathbb{E}[s_j],
\]
which proves the theorem in Case 2.

\subsection{Proof of Theorem \ref{thm:dipmark-feasibility}}

    We consider two exhaustive cases.

\medskip
\noindent\textit{Case 1: All tokens are green.}
If $s_{\iota[k]} = 1$ for all $k$, then $\mathbb{E}[s_j] = 1$ for all
$j \in \mathcal{V}$, so
$\arg\min_{j} \mathbb{E}[s_j] = \mathcal{V}$.
Any minimizer of $\hat{p}_j / p_j$ therefore belongs to $\mathcal{V}$,
and the inclusion holds trivially.

\medskip
\noindent\textit{Case 2: At least one token is red.}
Since the red prefix is non-empty, $\min_j s_j = 0$ and
$\min_j \mathbb{E}[s_j] = 0$. By the DipMark construction, the token
$\iota[0]$---the first in the shuffled order---has the highest probability
(over the random permutation) of landing in the red prefix, and therefore
attains the smallest expected score: $ \iota[0] \in \arg\min_{j \in \mathcal{V}} \mathbb{E}[s_j].$

In DipMark's reweighting procedure, probability mass is redistributed from
tokens at the start of the permuted list toward those at the end: the
earliest tokens are the first to have their probabilities reduced (and
potentially
zero), while later tokens remain invariant or receive
increased mass. Consequently, the maximum degradation $\hat{p}_j / p_j$ is
achieved by the token appearing earliest in the permuted order:
\[
     \iota[0] {\in} \arg\min_{j \in \mathcal{V}} \frac{\hat{p}_j}{p_j}.
\]
Combining both inclusions:
\[
    \iota[0] {\in} \arg\min_{j \in \mathcal{V}} \frac{\hat{p}_j}{p_j}
    \in \arg\min_{j \in \mathcal{V}} \mathbb{E}[s_j],
\]
which proves the theorem in Case 2.

\subsection{Proof of Theorem \ref{thm:its}}
\label{subsec:proof-its}

Let
$
 \xi_t^{(1)},\ldots,\xi_t^{(K)}
 \stackrel{\mathrm{i.i.d.}}{\sim}\mathcal L(\xi_t)
$
be auxiliary keys independent of the detector key $\xi_t$.
Define
\begin{align}
 \widehat p_{t,K}(j)
 &:=
 \frac{1}{K}
 \sum_{k=1}^{K}
 \mathbf{1}
 \left\{
  \Gamma(p_t,\xi_t^{(k)})=j
 \right\},
 \label{eq:its-empirical-distribution}\\
 R_{t,K}(j)
 &:=
 \frac{\widehat p_{t,K}(j)}{p_t(j)},
 \qquad j\in\mathcal s_j.
 \label{eq:its-empirical-ratio}
\end{align}
The attack selects
\begin{equation}
 x_t^\star
 \in
 \operatorname*{arg\,min}_{j\in\mathcal s_j}
 R_{t,K}(j).
 \label{eq:its-ratio-minimizer}
\end{equation}
For any candidate token $j$, write
\begin{equation}
 s(j)
 :=
 \left|
  \mu_t-\eta(\pi_t(j))
 \right|,
 \qquad
 \eta(r)=\frac{r-1}{|\mathcal V|-1}.
 \label{eq:its-candidate-score}
\end{equation}
Thus the detector contribution of token $j$ is
$-s(j)$.

\begin{theorem}[ITS detection-statistic reduction]

Let $v:=|\mathcal V|$. Suppose that $x_t^\star$ is
selected according to
Equation~\eqref{eq:its-ratio-minimizer} using auxiliary keys
independent of the detector key $\xi_t$. Then
\begin{equation}
 \mathbb{E}_{\xi_t}
 \left[
  s(x_t^\star)
 \right]
 =
 c_v,
 \qquad
 c_v:=\frac{2v-1}{6(v-1)}.
 \label{eq:its-attacked-null-score}
\end{equation}
The legitimate ITS token
$x_t=\Gamma(p_t,\xi_t)$ satisfies
\begin{equation}
 \mathbb{E}_{\xi_t}
 \left[
  s(x_t)
 \right]
 \leq c_v,
 \label{eq:its-legitimate-score-bound}
\end{equation}
with strict inequality whenever $|\mathcal s_j|\geq2$.
Consequently,
\begin{equation}
 \mathbb{E}_{\xi_t}
 \left[
  -s(x_t^\star)
 \right]
 <
 \mathbb{E}_{\xi_t}
 \left[
  -s(x_t)
 \right].
 \label{eq:its-proof-statistic-reduction}
\end{equation}
Thus empirical probability-ratio minimization reduces the
expected ITS detection statistic to its key-independent null
baseline.
\end{theorem}

\begin{proof}
Because $x_t^\star$ is a measurable function only of the
auxiliary keys
$\xi_t^{(1)},\ldots,\xi_t^{(K)}$, it is independent of the
detector key $\xi_t=(\mu_t,\pi_t)$. Condition on
$x_t^\star=j$. Since $\pi_t$ is a uniformly random
permutation, $\pi_t(j)$ is uniform on $\{1,\ldots,v\}$ and is
independent of $\mu_t$. Therefore,
\begin{align}
 \mathbb{E}_{\xi_t}
 \left[
  s(x_t^\star)
  \mid x_t^\star=j
 \right]
 &=
 \frac{1}{v}
 \sum_{r=1}^{v}
 \int_0^1
 \left|
  u-\eta(r)
 \right|du
 \nonumber\\
 &=
 \frac{1}{v}
 \sum_{r=1}^{v}
 \frac{
  \eta(r)^2+
  \bigl(1-\eta(r)\bigr)^2
 }{2}.
 \label{eq:its-null-score-integral}
\end{align}
For $\eta(r)=(r-1)/(v-1)$,
\begin{align}
 \frac{1}{v}\sum_{r=1}^{v}\eta(r)
 &=\frac12,\\
 \frac{1}{v}\sum_{r=1}^{v}\eta(r)^2
 &=\frac{2v-1}{6(v-1)}.
\end{align}
Substituting these identities into
Equation~\eqref{eq:its-null-score-integral} gives
\[
 \mathbb{E}_{\xi_t}
 \left[
  s(x_t^\star)
  \mid x_t^\star=j
 \right]
 =
 \frac{2v-1}{6(v-1)}
 =
 c_v.
\]
This value does not depend on $j$. Hence the law of total
expectation proves
Equation~\eqref{eq:its-attacked-null-score}.

We next analyze the legitimate ITS token. Define its
permuted rank and normalized rank by
\[
 R_t:=\pi_t(x_t),
 \qquad
 Z_t:=\eta(R_t).
\]
For each $r\in\{1,\ldots,v\}$, define the cumulative mass of
the first $r$ tokens in the permuted ordering:
\begin{equation}
 C_{t,r}^{\pi_t}
 :=
 \sum_{j:\pi_t(j)\leq r}p_t(j).
 \label{eq:its-cumulative-mass}
\end{equation}
By the inverse-CDF decoder,
\begin{equation}
 R_t\leq r
 \quad\Longleftrightarrow\quad
 \mu_t\leq C_{t,r}^{\pi_t}.
 \label{eq:its-rank-cdf-equivalence}
\end{equation}

We first establish the marginal distribution of $R_t$.
Conditioned on $\pi_t$, the probability that the output has
rank $r$ equals the probability mass of the token occupying
that rank. Thus,
\begin{align}
 \Pr(R_t=r)
 &=
 \mathbb{E}_{\pi_t}
 \left[
  p_t\bigl(\pi_t^{-1}(r)\bigr)
 \right]
 \nonumber\\
 &=
 \frac{1}{v}
 \sum_{j\in\mathcal V}p_t(j)
 =
 \frac{1}{v}.
 \label{eq:its-uniform-output-rank}
\end{align}
Therefore, $R_t$ is uniform on $\{1,\ldots,v\}$ and $Z_t$ is
uniform on
\[
 \left\{
  0,\frac{1}{v-1},\ldots,1
 \right\}.
\]

Although $Z_t$ has the same marginal distribution as the rank
of a key-independent token, it is positively dependent on
$\mu_t$. Indeed, every token belongs to the first $r$
positions of a uniform permutation with probability $r/v$.
Consequently,
\begin{align}
 \mathbb{E}_{\pi_t}
 \left[
  C_{t,r}^{\pi_t}
 \right]
 &=
 \sum_{j\in\mathcal V}
 p_t(j)
 \Pr\bigl(\pi_t(j)\leq r\bigr)
 \nonumber\\
 &=
 \frac{r}{v}
 \sum_{j\in\mathcal V}p_t(j)
 =
 \frac{r}{v}.
 \label{eq:its-expected-cumulative-mass}
\end{align}
For every $u\in[0,1]$, Equation
\eqref{eq:its-rank-cdf-equivalence} gives
\begin{align}
 \Pr(\mu_t\leq u,R_t\leq r)
 &=
 \mathbb{E}_{\pi_t}
 \left[
  \Pr
  \left(
   \mu_t\leq
   \min\{u,C_{t,r}^{\pi_t}\}
   \mid\pi_t
  \right)
 \right]
 \nonumber\\
 &=
 \mathbb{E}_{\pi_t}
 \left[
  \min\{u,C_{t,r}^{\pi_t}\}
 \right]
 \nonumber\\
 &\geq
 u\,
 \mathbb{E}_{\pi_t}
 \left[
  C_{t,r}^{\pi_t}
 \right]
 \nonumber\\
 &=
 u\frac{r}{v}.
 \label{eq:its-positive-quadrant-dependence}
\end{align}
The inequality uses
\[
 \min\{u,c\}\geq uc,
 \qquad u,c\in[0,1].
\]
Since $\Pr(\mu_t\leq u)=u$ and
$\Pr(R_t\leq r)=r/v$, Equation
\eqref{eq:its-positive-quadrant-dependence} proves that
$\mu_t$ and $R_t$, and hence $\mu_t$ and $Z_t$, are
positively quadrant dependent.

Let $\widetilde Z_t$ be independent of $\mu_t$ and have the
same marginal distribution as $Z_t$. Positive quadrant
dependence implies
\begin{equation}
 \Pr(\mu_t>u,Z_t>u)
 \geq
 \Pr(\mu_t>u)
 \Pr(\widetilde Z_t>u)
 \label{eq:its-upper-tail-dependence}
\end{equation}
for every $u\in[0,1]$. Using the layer-cake representation,
\begin{align}
 \mathbb{E}
 \left[
  \min\{\mu_t,Z_t\}
 \right]
 &=
 \int_0^1
 \Pr(\mu_t>u,Z_t>u)\,du
 \nonumber\\
 &\geq
 \int_0^1
 \Pr(\mu_t>u)
 \Pr(\widetilde Z_t>u)\,du
 \nonumber\\
 &=
 \mathbb{E}
 \left[
  \min\{\mu_t,\widetilde Z_t\}
 \right].
 \label{eq:its-minimum-dominance}
\end{align}

Both $Z_t$ and $\widetilde Z_t$ have mean $1/2$. Therefore,
using $|a-b|=a+b-2\min\{a,b\}$,
\begin{align}
 \mathbb{E}
 \left[
  |\mu_t-Z_t|
 \right]
 &=
 1-
 2\mathbb{E}
 \left[
  \min\{\mu_t,Z_t\}
 \right]
 \nonumber\\
 &\leq
 1-
 2\mathbb{E}
 \left[
  \min\{\mu_t,\widetilde Z_t\}
 \right]
 \nonumber\\
 &=
 \mathbb{E}
 \left[
  |\mu_t-\widetilde Z_t|
 \right]
 =
 c_v.
 \label{eq:its-distance-comparison}
\end{align}
Because
$s(X_t)=|\mu_t-Z_t|$, this proves
Equation~\eqref{eq:its-legitimate-score-bound}.

It remains to establish strictness. If
$|\mathcal s_j|\geq2$, every token in $\mathcal s_j$ has
probability strictly between zero and one. Hence
\[
 \Pr_{\pi_t}
 \left(
  0<C_{t,1}^{\pi_t}<1
 \right)
 >0.
\]
For every $u\in(0,1)$ and every $c\in(0,1)$,
\[
 \min\{u,c\}>uc.
\]
Thus Equation~\eqref{eq:its-positive-quadrant-dependence} is
strict for $r=1$ and $u\in(0,1)$. In particular, for
\[
 0<u<\frac{1}{v-1},
\]
the event $Z_t\leq u$ is equivalent to $R_t=1$, and
Equation~\eqref{eq:its-upper-tail-dependence} is strict on an
interval of positive measure. It follows that
\[
 \mathbb{E}
 \left[
  \min\{\mu_t,Z_t\}
 \right]
 >
 \mathbb{E}
 \left[
  \min\{\mu_t,\widetilde Z_t\}
 \right],
\]
and therefore
\[
 \mathbb{E}_{\xi_t}
 \left[
  s(X_t)
 \right]
 <
 c_v
 =
 \mathbb{E}_{\xi_t}
 \left[
  s(x_t^\star)
 \right].
\]
Multiplying by $-1$ gives
\[
 \mathbb{E}_{\xi_t}
 \left[
  -s(x_t^\star)
 \right]
 <
 \mathbb{E}_{\xi_t}
 \left[
  -s(X_t)
 \right],
\]
which proves
Equation~\eqref{eq:its-proof-statistic-reduction}.
\end{proof}

\subsection{Proof of Theorem \ref{thm:synthid-layerwise-alignment}}
Fix a tournament layer $\ell$ and let
$p^{(\ell-1)}$ be its incoming token distribution. For an
ordered score value $a$, define
\begin{align}
 Q_{\ell,<}(a)
 :=
 \sum_{k:g_{\ell,k}<a}p_k^{(\ell-1)}, \quad
 Q_{\ell,\leq}(a)
 :=
 \sum_{k:g_{\ell,k}\leq a}p_k^{(\ell-1)}, \quad
 Q_{\ell,=}(a)
 :=
 Q_{\ell,\leq}(a)-Q_{\ell,<}(a).
\end{align}
For every supported token $i$, we have
$Q_{\ell,=}(g_{\ell,i})>0$.

Consider a match containing $N$ independent candidates drawn
from $p^{(\ell-1)}$. The probability that the largest
$g$-value among these candidates is exactly $a$ is
\[
 Q_{\ell,\leq}(a)^N-Q_{\ell,<}(a)^N.
\]
Conditioned on the largest value being $a$, uniform
tie-breaking and exchangeability imply that the winning token
is distributed proportionally to $p^{(\ell-1)}$ among the
tokens whose $g$-value equals $a$. Consequently,
\begin{equation}
 p_i^{(\ell)}
 =
 p_i^{(\ell-1)}
 \frac{
 Q_{\ell,\leq}(g_{\ell,i})^N
 -
 Q_{\ell,<}(g_{\ell,i})^N
 }{
 Q_{\ell,=}(g_{\ell,i})
 }.
 \label{eq:synthid-layer-update-proof}
\end{equation}
Therefore, the layerwise probability ratio is
\begin{align}
 R_{\ell,i}
 &:=
 \frac{p_i^{(\ell)}}{p_i^{(\ell-1)}}\\
 &=
 \frac{
 Q_{\ell,\leq}(g_{\ell,i})^N
 -
 Q_{\ell,<}(g_{\ell,i})^N
 }{
 Q_{\ell,\leq}(g_{\ell,i})
 -
 Q_{\ell,<}(g_{\ell,i})
 }\\
 &=
 \sum_{r=0}^{N-1}
 Q_{\ell,\leq}(g_{\ell,i})^{N-1-r}
 Q_{\ell,<}(g_{\ell,i})^r,
 \label{eq:synthid-ratio-polynomial}
\end{align}
where the last equality follows from
\[
 \frac{a^N-b^N}{a-b}
 =
 \sum_{r=0}^{N-1}a^{N-1-r}b^r.
\]

Now consider $i,j\in\mathcal S$ satisfying
$g_{\ell,i}<g_{\ell,j}$. Because both score levels have
positive incoming probability mass,
\begin{equation}
 Q_{\ell,<}(g_{\ell,i})
 <
 Q_{\ell,\leq}(g_{\ell,i})
 \leq
 Q_{\ell,<}(g_{\ell,j})
 <
 Q_{\ell,\leq}(g_{\ell,j}).
 \label{eq:synthid-cdf-order}
\end{equation}
Every term in Equation~\eqref{eq:synthid-ratio-polynomial}
is nondecreasing in both cumulative masses. Moreover, because
$N\geq2$,
\[
 Q_{\ell,\leq}(g_{\ell,i})^{N-1}
 <
 Q_{\ell,\leq}(g_{\ell,j})^{N-1}.
\]
Thus at least the $r=0$ term is strictly larger for token $j$,
while none of the other terms is smaller. It follows that
\[
 g_{\ell,i}<g_{\ell,j}
 \quad\Longrightarrow\quad
 R_{\ell,i}<R_{\ell,j}.
\]

Hence $R_{\ell,i}$ is strictly increasing in the ordered
$g$-value. All tokens attaining the smallest $g$-value attain
the smallest probability ratio, and every token with a larger
$g$-value has a strictly larger ratio. Therefore,
\[
 \arg\min_{i\in\mathcal S}R_{\ell,i}
 =
 \arg\min_{i\in\mathcal S}g_{\ell,i}.
\]
Since $g_{\ell,i}$ is the contribution of token $i$ to the
mean detector at layer $\ell$, minimum-ratio selection also
minimizes that layer's mean-score contribution.

\subsection{Proof of Theorem \ref{thm:synthid-bayesian-layerwise}}

We first prove that the conditional Bayesian factor is
nondecreasing in the current $g$-value.

Let
$
    F_g^{-}(z):=\Pr(G<z)=F_g(z)-f_g(z).
$
Conditioned on $\psi_\ell=c$, the winning $g$-value is the
maximum of $c$ i.i.d. null $g$-values. Therefore,
$    \Pr(G_{(c)}=z)
    =
    F_g(z)^c-F_g^{-}(z)^c.
$
Consequently,
\begin{align}
    \Lambda_c(z)
    &=
    \frac{F_g(z)^c-F_g^{-}(z)^c}
         {F_g(z)-F_g^{-}(z)}
    =
    \sum_{k=0}^{c-1}
    F_g(z)^{c-1-k}F_g^{-}(z)^k,
    \label{eq:bayes-factor-discrete}
\end{align}
where the second equality follows from
\[
    \frac{a^c-b^c}{a-b}
    =
    \sum_{k=0}^{c-1}a^{c-1-k}b^k.
\]

If $z_1<z_2$, then
$
    F_g(z_1)\leq F_g(z_2),
    F_g^{-}(z_1)\leq F_g^{-}(z_2).
$

Every term in Equation~\eqref{eq:bayes-factor-discrete} is
nondecreasing in both arguments. Hence
\[
    z_1<z_2
    \quad\Longrightarrow\quad
    \Lambda_c(z_1)\leq\Lambda_c(z_2).
\]
Because $q_c(u)\geq0$ and $\sum_c q_c(u)=1$, it follows that
\[
    \textrm{BS}_\ell(z_1;u)
    =
    \sum_{c=1}^{N}q_c(u)\Lambda_c(z_1)
    \leq
    \sum_{c=1}^{N}q_c(u)\Lambda_c(z_2)
    =
    \textrm{BS}_\ell(z_2;u).
\]
Thus $\textrm{BS}_\ell(z;u)$ is nondecreasing in $z$.

For completeness, when the $g$-value distribution is
continuous,
\[
    \Pr(G_{(c)}\in dz)
    =
    cF_g(z)^{c-1}f_g(z)\,dz,
\]
and therefore
\[
    \Lambda_c(z)=cF_g(z)^{c-1},
\]
which is again nondecreasing in $z$.

We next consider the Tournament-sampling probability ratio.
For a score level $z$, define
\[
    Q_{\ell,<}(z)
    :=
    \sum_{j:g_{\ell,j}<z}p_j^{(\ell-1)},
    \qquad
    Q_{\ell,\leq}(z)
    :=
    \sum_{j:g_{\ell,j}\leq z}p_j^{(\ell-1)},
\]
and
\[
    Q_{\ell,=}(z)
    :=
    Q_{\ell,\leq}(z)-Q_{\ell,<}(z).
\]
The Tournament-sampling update gives
\[
    R_{\ell,i}
    =
    \frac{
      Q_{\ell,\leq}(g_{\ell,i})^N
      -
      Q_{\ell,<}(g_{\ell,i})^N
    }{
      Q_{\ell,=}(g_{\ell,i})
    }.
\]
Using the difference-of-powers identity,
\begin{align}
    R_{\ell,i}
    &=
    \sum_{k=0}^{N-1}
    Q_{\ell,\leq}(g_{\ell,i})^{N-1-k}
    Q_{\ell,<}(g_{\ell,i})^k.
    \label{eq:tournament-ratio-polynomial}
\end{align}

Both $Q_{\ell,\leq}(z)$ and $Q_{\ell,<}(z)$ are
nondecreasing in $z$. Since every coefficient in
Equation~\eqref{eq:tournament-ratio-polynomial} is
nonnegative,
\[
    g_{\ell,i}<g_{\ell,j}
    \quad\Longrightarrow\quad
    R_{\ell,i}\leq R_{\ell,j}.
\]
For distinct supported score levels and $N\geq2$, the
inequality is strict. Therefore,
\[
    \arg\min_i R_{\ell,i}
    =
    \arg\min_j g_{\ell,j}.
\]

We have shown that both the layerwise probability multiplier
and the conditional Bayesian factor are minimized at the
smallest current-layer $g$-value. Hence,
\[
    \arg\min_i R_{\ell,i}
    =
    \arg\min_i g_{\ell,i}
    \subseteq
    \arg\min_i \textrm{BS}_\ell(g_{\ell,i};u).
\]

Finally, if $q_c(u)>0$ for at least one $c\geq2$, then
$\Lambda_c(z)$, and consequently $\textrm{BS}_\ell(z;u)$, is strictly
increasing over the supported score levels. In that case,
\[
    \arg\min_i R_{\ell,i}
    =
    \arg\min_i g_{\ell,i}
    =
    \arg\min_i \textrm{BS}_\ell(g_{\ell,i};u).
\]

\section{Proof of Robustness-Detectability Tradeoff}
\label{app:proof-tradeoff}

\subsection{Proof of Theorem \ref{thm:defeatable}}

    Consider an adversary that, at each decoding step $t$, selects the token
minimizing the probability ratio:
$
    j^* = \arg\min_{j \in \mathcal{V}} \frac{\hat{p}_j}{p_j}.
$
By the defeatability condition (Definition~\ref{def:defeatable}), this
selection is equivalent to
$
    j^* = \arg\min_{j \in \mathcal{V}} s_j.
$
The Cumulative detection score over a sequence of length $T$ is $S = \sum_{t=1}^{T} s_{i_t^*}$.

Since the adversary selects the score-minimizing token at every step, $S$
is minimized in expectation over all possible generation paths. The
detector declares watermark presence by thresholding $S$ against a fixed
value $\tau$; the true positive rate is
$\Pr_{\hat{p}}\!\left(S \geq \tau\right).$
Minimizing $S$ at each step reduces $\mathbb{E}_{\hat{p}}[S]$, which
directly reduces $\mathrm{TPR}$. Therefore, the adversary can drive the
detection rate to baseline levels, and the watermark is not robust.

\subsection{Proof of Theorem \ref{thm:tradeoff}}

    By the Neyman--Pearson lemma ~\cite{Neyman1933OnTP}, the most powerful
test at significance level $\alpha$ for the simple hypothesis pair
$H_0 : x \sim p$ versus $H_1 : x \sim \hat{p}$ is the likelihood-ratio
test, which rejects $H_0$ when $\Lambda(x) = \hat{p}(x)/p(x)$ exceeds a
threshold $\tau_\alpha$ calibrated to achieve false positive rate $\alpha$.
No other test achieves a higher true positive rate at the same false
positive rate.

Detection power is therefore maximized when the score function $s$ used
to compute the detection statistic preserves the ordering of
$\Lambda(x)$---that is, when $s(x) \propto \log(\hat{p}(x)/p(x))$. Under
this assignment, tokens with larger likelihood ratios receive higher scores.

However, aligning $s$ with $\log(\hat{p}/p)$ implies that
\[
    \arg\min_{j}\, s_j
    = \arg\min_{j}\, \log\frac{\hat{p}_j}{p_j}
    = \arg\min_{j}\, \frac{\hat{p}_j}{p_j}.
\]
This is precisely the defeatability condition of
Definition~\ref{def:defeatable}. By Theorem~\ref{thm:defeatable}, the
scheme is therefore not robust. Since both conclusions follow from the same
alignment requirement, improving detectability necessarily strengthens the
structure that the adversary exploits.

\section{Proof of Theoretical Attack Effectiveness}
\label{app:proof-effectiveness}

\begin{lemma}[Minimum versus convex combination]
\label{lem:min-convex-combination}
For arbitrary real numbers $a_1,\ldots,a_n$ and weights
$q_1,\ldots,q_n$ satisfying
$
 q_j\geq 0
 \text{ and }
 \sum_{i=1}^{n}q_j=1,
$
we have
$
 \min_{1\leq i\leq n}a_j
 \leq
 \sum_{i=1}^{n}q_j a_j
 \label{eq:min-convex-combination}
$.
Equality holds if and only if every index assigned positive
weight is a minimizer of $a_j$.
\end{lemma}

\emph{Proof:} Let $a_{\min}=\min_j a_j$. Since $a_{\min}\leq a_j$ for every
$i$, multiplying each inequality by the nonnegative weight
$q_j$ and summing gives
\[
 \sum_{i=1}^{n}q_j a_{\min}
 \leq
 \sum_{i=1}^{n}q_j a_j.
\]
Because $\sum_jq_j=1$, the left-hand side equals
$a_{\min}$, proving~\eqref{eq:min-convex-combination}.
The inequality is strict exactly when $q_j>0$ for at least one
index satisfying $a_j>a_{\min}$.

\subsection{Proof of Theorem \ref{thm:minimum-ratio-effectiveness}}

Fix a decoding step $t$ and condition on
$h_{1:t-1}=h$ and $\xi_t$. By
Equation~\eqref{eq:effectiveness-alignment}, minimum-ratio
selection also minimizes the token score:
\[
 s_t(x_t^{\mathrm A}\mid h,\xi_t)
 =
 \min_{i\in\mathcal{S}_t(h)}
 s_t(i\mid h,\xi_t).
\]

The conditional expected score of an unwatermarked token is
\[
 \mathbb{E}
 \left[
   s_t(x_t^0)
   \mid h_{1:t-1}=h,\xi_t
 \right]
 =
 \sum_{i\in\mathcal{S}_t(h)}
 p_t(i\mid h)s_t(i\mid h,\xi_t).
\]
Because the values $p_t(i\mid h)$ are nonnegative and sum to
one, Lemma~\ref{lem:min-convex-combination} gives
\[
 \min_{i\in\mathcal{S}_t(h)}
 s_t(i\mid h,\xi_t)
 \leq
 \sum_{i\in\mathcal{S}_t(h)}
 p_t(i\mid h)s_t(i\mid h,\xi_t).
\]
This proves~\eqref{eq:per-step-expectation-bound}.

Multiplying by $w_t\geq0$, summing over $t$, and applying the
law of total expectation gives
Under fresh-key null invariance, the counterfactual null scores
have the same joint distribution as the scores generated under
$H_0$. Hence the right-hand side equals
$\mathbb{E}[s_j^0]$, proving
Equation~\eqref{eq:cumulative-expectation-bound}.

\subsection{Proof of Theorem \ref{thm:stochastic-dominance}}

We construct the attacked and null score sequences on a common
probability space.

Suppose the attacked prefix
$h_{1:t-1}$ has already been generated. Conditional
on this prefix and the current key realization $\xi_t$, the
adversary selects $x_t^{\mathrm A}$ by minimizing the probability
ratio. By the pointwise alignment condition
in~\eqref{eq:pointwise-alignment},
\[
 \phi(\xi_t,h_{1:t-1},x_t^{\mathrm A})
 =
 \min_{i:p_t(i\mid h_{1:t-1})>0}
 \phi(\xi_t,h_{1:t-1},i).
\]
Define the attacked score as
\[
 Z_t^{\mathrm A}
 :=
 \phi(\xi_t,h_{1:t-1},x_t^{\mathrm A}).
\]

On the same probability space, independently sample the
counterfactual null token
\[
 Y_t\sim p_t(\cdot\mid h_{1:t-1})
\]
and define
\[
 \widetilde Z_t^0
 :=
 \phi(\xi_t,h_{1:t-1},Y_t).
\]
Because $Y_t$ belongs to the support of the base distribution,
the score-minimizing property gives
\[
    Z_t^{\mathrm A}\leq\widetilde Z_t^0
    \qquad\text{almost surely}.
 \label{eq:per-step-coupling}
\]

Although $\widetilde Z_t^0$ is evaluated at the attacked prefix,
fresh-key null invariance implies
\[
 \mathcal{L}
 \left(
   \widetilde Z_t^0
   \mid h_{1:t-1}
 \right)
 =
 G_t.
\]
Therefore, its conditional distribution is exactly the same as
that of a null score used to calibrate the detector and does not
depend on the attacked history. Repeating the construction for
$t=1,\ldots,T$ yields
\[
 (\widetilde Z_1^0,\ldots,\widetilde Z_T^0)
 \stackrel{d}{=}
 (Z_1^0,\ldots,Z_T^0).
\]
At the same time, Equation~\eqref{eq:per-step-coupling} holds
simultaneously at every step:
\[
    Z_t^{\mathrm A}\leq\widetilde Z_t^0,
    \qquad t=1,\ldots,T,
    \quad\text{almost surely}.
\]

Since $\Phi_T$ is coordinatewise nondecreasing,
\[
 \begin{aligned}
 D_T^{\mathrm A}
   &=
   \Phi_T(Z_1^{\mathrm A},\ldots,Z_T^{\mathrm A})
   \leq
   \Phi_T(\widetilde Z_1^0,\ldots,\widetilde Z_T^0)
   =:\widetilde D_T^0
   \qquad\text{almost surely}.
 \end{aligned}
\]
Moreover, $\widetilde D_T^0\stackrel{d}{=}D_T^0$. Hence, for
every threshold $\tau$,
\[
 \Pr_{\mathrm A}(D_T^{\mathrm A}\geq\tau)
 \leq
 \Pr(\widetilde D_T^0\geq\tau)
 =
 \Pr_0(D_T^0\geq\tau).
\]
Thus
$D_T^{\mathrm A}\preceq_{\mathrm{st}}D_T^0$.
Setting $\tau=\tau_\varepsilon$ gives
\[
 \operatorname{TPR}_{\mathrm A}(\varepsilon)
 \leq
 \Pr_0(D_T^0\geq\tau_\varepsilon)
 \leq\varepsilon,
\]
which proves~\eqref{eq:attacked-tpr-bound}.

\subsection{Proof of Theorem \ref{cor:practical-tpr}}

Let $E_t$ be the event that the estimator selects a
score-minimizing token within $C_t$. If $x_t^0\in C_t$ and
$E_t$ occurs, then
\[
 s_t(x_t^{\mathrm A})
 =
 \min_{i\in C_t}s_t(i)
 \leq
 s_t(x_t^0).
\]
Therefore, the coordinatewise coupling used in
Theorem~\ref{thm:stochastic-dominance} fails only if
$x_t^0\notin C_t$ or $E_t^c$ occurs. Hence,
\[
 \{D_T^{\mathrm A}\geq\tau_\varepsilon\}
 \subseteq
 \{D_T^0\geq\tau_\varepsilon\}
 \cup
 \bigcup_{t=1}^{T}
 \left(
   \{x_t^0\notin C_t\}\cup E_t^c
 \right).
\]
Taking probabilities and applying the union bound gives
\[
 \Pr_{\mathrm A}
 (D_T^{\mathrm A}\geq\tau_\varepsilon)
 \leq
 \varepsilon+\sum_{t=1}^{T}(\rho_t+\zeta_t),
\]
which proves~\eqref{eq:practical-tpr-bound}.